\documentclass[preprint]{elsarticle}

\usepackage{lineno,hyperref}
\modulolinenumbers[5]

\usepackage{geometry}
\usepackage{pifont}
\usepackage{amscd}
\usepackage{amsmath,amsfonts,amssymb,amsthm}
\usepackage{latexsym}
\usepackage{stmaryrd}
\usepackage{overpic}
\usepackage{extarrows}
\usepackage{subfig}
\usepackage{multirow}
\usepackage{hyperref}
\usepackage{color}

\journal{Elsevier}

\newtheorem{theorem}{Theorem}[section]
\newtheorem{lemma}[theorem]{Lemma}
\newtheorem{definition}{Definition}[section]
\newdefinition{remark}{Remark}[section]

\numberwithin{equation}{section} 
\numberwithin{table}{section}
\numberwithin{figure}{section}

\newcommand\od[2]{\dfrac{\rd {#1}}{\rd {#2}}}
\newcommand\pd[2]{\dfrac{\partial {#1}}{\partial {#2}}}
\newcommand\rd{\mathrm{d}} 
\newcommand\bA{\boldsymbol{A}}
\newcommand\bB{\boldsymbol{B}}
\newcommand\bP{\boldsymbol{P}}
\newcommand\bM{\boldsymbol{M}}
\newcommand\bH{\boldsymbol{H}}
\newcommand\bT{\boldsymbol{T}}
\newcommand\bI{\boldsymbol{I}}
\newcommand\bX{\boldsymbol{X}}
\newcommand\bi{\boldsymbol{i}}
\newcommand\bj{\boldsymbol{j}}
\newcommand\bW{\boldsymbol{W}}
\newcommand\bsigma{\boldsymbol{\sigma}}
\newcommand\bLambda{\boldsymbol{\Lambda}}
\newcommand\R{\mathbb R}
\newcommand\mS{\mathbb S}
\newcommand\cO{\mathcal O}
\newcommand\cM{\mathcal{M}}
\newcommand\cA{\mathcal {A}}
\newcommand{\grad}{{\rm grad}}

\begin{document}

\begin{frontmatter}



\title{ Multiphase Allen-Cahn and Cahn-Hilliard Models and Their
Discretizations with the Effect of Pairwise Surface
Tensions\tnoteref{Funding}}
\tnotetext[Funding]{
This study was partially supported by DOE Grant DE-SC0009249 as part
of the Collaboratory on Mathematics for Mesoscopic Modeling of
Materials.}


\author[psu]{Shuonan Wu}
\ead{wsn1987@gmail.com}
  
\author[psu]{Jinchao Xu\corref{cor1}}
\ead{xu@math.psu.edu}

\address[psu]{Department of Mathematics, Pennsylvania State University,
University Park, PA, 16802, USA}

\cortext[cor1]{Corresponding author}

\begin{abstract}
In this paper, the mathematical properties and numerical
discretizations of multiphase models that simulate the phase
separation of an $N$-component mixture are studied. For the general
choice of phase variables, the unisolvent property of the coefficient
matrix involved in the $N$-phase models based on the pairwise surface
tensions is established. Moreover, the symmetric positive-definite
property of the coefficient matrix on an $(N-1)$-dimensional
hyperplane --- which is of fundamental importance to the
well-posedness of the models --- can be proved equivalent to some
physical condition for pairwise surface tensions. The $N$-phase
Allen-Cahn and $N$-phase Cahn-Hilliard equations can then be derived
from the free-energy functional. A natural property is that the
resulting dynamics of concentrations are independent of phase
variables chosen.  Finite element discretizations for $N$-phase models
can be obtained as a natural extension of the existing discretizations
for the two-phase model. The discrete energy law of the numerical
schemes can be proved and numerically observed under some restrictions
pertaining to time step size. Numerical experiments including the
spinodal decomposition and the evolution of triple junctions are
described in order to investigate the effect of pairwise surface
tensions. 
\end{abstract}

\begin{keyword}
Multiphase \sep Allen-Cahn \sep Cahn-Hilliard \sep pairwise surface
tensions
\end{keyword}

\end{frontmatter}



\section{Introduction} \label{sec:intro}

Multiphase flows are frequently encountered in biomedical, chemical,
and engineering applications. The dynamics of multiphase flows
associate with a wide range of fundamental physical properties such as
pairwise surface tensions, wetting spreading, and formating contact
angles among multiple materials \cite{de2013capillarity}. On the other
hand, multiphase flows are challenging from the points of view of both
mathematical modeling and numerical methods due to the complexity of
the moving interface. 

There are two main approaches to moving interface problems: the direct
approach and the indirect approach. The direct approach obtains
information pertaining to the interface by tracking quantities
associated with it. Therefore, the direct approach relies on the
parameterization method \cite{zabusky1983regularization}, the immersed
boundary method \cite{peskin1977numerical}, the volume-of-fluid method
\cite{hirt1981volume}, and/or the front tracking method
\cite{unverdi1992front}. It is known that the direct approach commonly
encounters difficulty handling topological changes, such as pinches,
splits, and merging --- all of which can be handled easily by the
indirect approach. The level set method \cite{osher1988fronts} and the
phase field method \cite{rayleigh1892xx} are both examples of popular
indirect methods. In this paper, however, we focus on the phase field
method for modeling the effect of pairwise surface tensions for
$N$-phase flows ($N\geq 2$).

With the phase field method, the thickness of the sharp interface
between the two phases is supposed to be very small but positive. The
state of the system is then represented by a set of smooth functions
called phase variables or order parameters. The evolution of the
system is driven by the gradient of a total free-energy, which is
the sum of two terms: a bulk free-energy term, whose effect tends to
separate the flows, and a capillary term, whose effect tends to mix the
flows. The capillary term depends on the gradient of the order
parameters, which accounts for the energy of the interfacial tensions
between flows.

Drawing on the large body of research on two-phase flows
\cite{anderson1998diffuse, lowengrub1998quasi, chen2002phase,
badalassi2003computation, yue2004diffuse, shen2010numerical},
researchers have produced many theoretical and numerical studies on
three-phase flows involving the effect of pairwise surface tensions
\cite{kim2004conservative, kim2005phase, boyer2006study,
kim2007phase, boyer2011numerical}. In these models,
the given pairwise surface tensions $\sigma_{ij}$ are decomposed
into three positive phase-specific surface-tension coefficients as 
$$ 
\sigma_{12} = \sigma_1 + \sigma_2, \quad \sigma_{13} = \sigma_1 +
\sigma_3, \quad \sigma_{23} = \sigma_2 + \sigma_3,
$$ 
whose existence is equivalent to the triangle inequality of the
pairwise surface tensions. However, this decomposition encounters
difficulties for cases in which $N\geq 4$, as the number of pairwise
tensions would be greater than the number of phase-specific
surface-tension coefficients, which leads to an overdetermined system
\cite{kim2009generalized, kim2012phase}. In \cite{kim2007phase}, a
phenomenological continuum surface tension force was introduced by
coupling Navier-Stokes equations through the mean curvature of the
interface. Further, the generalization of this approach to an
arbitrary number of phases with the purpose of avoiding the
solvability issue was discussed in \cite{kim2009generalized}.

Generalizations of diffuse models to an arbitrary number of phases
have recently been introduced and studied. In most of the existing
models for multiphase flows, the phase variables are chosen
specifically as concentrations of mixture $c_i$, whose sum is equal to
$1$. Examples of such models include $N$-phase Allen-Cahn equations
\cite{kornhuber2006robust, lee2012efficient} and $N$-phase
Cahn-Hilliard equations \cite{barrett1997finite, barrett1998finite,
lee2008second, vanherpe2010multigrid, lee2012practically,
graser2014nonsmooth}. An benefit of these models is that their
consistency with the two-phase model can be easily proven.  However,
the pairwise surface tensions are not involved in the energy-density
function so that the homogeneous surface tensions are implied in
most of the existing models intrinsically. As the physical
concentrations must belong to the $(N-1)$-dimensional Gibbs simplex
\cite{porter2009phase}, a variable Lagrangian multiplier should be
introduced in the dynamic equations.

In order to incorporate the pairwise surface tensions into the phase
field model, several generalized models have been proposed based on
the generalized total free-energy functional. In
\cite{elliott1991generalised}, Elliott and Luckhaus set the total
free-energy functional as 
$$ 
\mathcal{E}(\vec{c}) := \int_\Omega \left[ \Psi(\vec{c}) +
\frac{1}{2} 
(\boldsymbol{\Gamma} \nabla \vec{c}) : \nabla \vec{c}  
\right],
$$ 
where $\boldsymbol{\Gamma}$ is the $N\times N$ symmetric-positive
semi-definite matrix, i.e. a symmetric coefficient matrix is
introduced in the capillary-energy term. They also gave a global
existence result under constant mobility when $\boldsymbol{\Gamma}=
\gamma \bI$. Eyre \cite{eyre1998unconditionally} then studied this
system and determined its equilibrium and dynamic behavior. Recently,
Boyer and Minjeaud \cite{boyer2014hierarchy} proposed a generalization
of the well-known two-phase Cahn-Hilliard model for the modeling of
$N$-phase mixtures using the concentrations as the phase variables.
Dong \cite{dong2014efficient} established an algebraic relationship
between the coefficient matrix and the pairwise surface tensions under
a special choice of phase variables and gave the coupled system
between the phase field and Navier-Stokes equations in the
thermodynamics framework \cite{heida2012development}. One main feature
of these works is that, thanks to a relevant choice of free-energy,
the model coincides exactly with the two-phase model.  Dong then
derived a formulation for the general phase variables in
\cite{dong2015physical} by eliminating one variable in order to relax
the algebraic relationship.   

In this paper, we begin by applying the general phase variables
without eliminating any of the variables, and we rebuild the
relationship between the coefficient matrix and the pairwise surface
tensions in a compact form.  By drawing on a recent work on the close
connection between the symmetric matrix space and simplex
\cite{hu2015finite}, we obtain the unisolvent property of the
coefficient matrix on the tangent space of the solution manifold.
Furthermore, the symmetric positive-definite (SPD) property on the
tangent space proposed as an open problem in \cite{boyer2014hierarchy,
dong2014efficient, dong2015physical}, is answered by two equivalent
conditions from both the algebraic and the geometric point of view,
see Theorem \ref{thm:SPD}. We note that this property is fundamentally
important to the well-posedness of the dynamic system. This is the
first major contribution of the present study to the field. 

The second principle contribution of the present study is the
derivation of the $N$-phase Allen-Cahn and Cahn-Hilliard equations
under the generalized total free-energy functional. As the gradient
flow on the solution manifold, the Allen-Cahn equations make sense
only under the given inner product on the tangent space in
energy-variation framework. Here, we apply the inner product on the
tangent space induced from the choice of generalized phase variables,
so that the dynamics of the concentrations are independent of the
choice of phase variables. A similar technique can be applied to
$N$-phase Cahn-Hilliard equations to obtain the same property. When 
$N$-phase Allen-Cahn and Cahn-Hilliard equations are written in a strong
formulation, the orthogonal projection to the tangent space will
naturally translate into the variable Lagrangian multiplier as shown
in models reported in \cite{ kornhuber2006robust, lee2012efficient,
lee2008second, vanherpe2010multigrid, lee2012practically,
graser2014nonsmooth}. This implies that our models can be viewed as
a natural extension of the existing models while accounting for and
including effect of pairwise surface tensions on the multiphase flows. 

Based on the above properties, we propose finite element
discretizations for $N$-phase models. The semi-implicit,
fully-implicit, and modified Crank-Nicolson scheme, are considered for
$N$-phase Allen-Cahn equations, and the semi-implicit, fully-implicit,
and modified Crank-Nicolson scheme, are considered for $N$-phase
Cahn-Hilliard equations.  Each of these schemes can be viewed as a
natural extension of the existing numerical schemes for two-phase
flows \cite{du1991numerical, shen2010numerical}.  The discrete energy
law of the numerical schemes is also discussed.  

The rest of this paper is organized as follows. In Section
\ref{sec:N-model}, we consider the generalized phase variables and the
free-energy functional with a coefficient matrix in the capillary
term.  The solvability and SPD property of the coefficient matrix are
discussed. We also derive $N$-phase Allen-Cahn and Cahn-Hilliard
equations so that the corresponding dynamics of concentrations are
independent of the choice of phase variables. In Section
\ref{sec:discretization}, the finite element discretizations of the
$N$-phase models are described and energy stability of each is
considered. Numerical experiments showing the effect of the pairwise
tensions on the multiphase flows and the accuracy of the schemes are
presented in Section \ref{sec:tests}. Some closing remarks are given
in Section \ref{sec:remarks}.

\section{$N$-phase Models} \label{sec:N-model}

First, we introduce some notation that will be used throughout this
paper. Let $\Omega \subset \R^{d}~(d=2,3)$ be the bounded domain, and
$\partial \Omega$ the domain boundary. The unit outer normal vector of
$\partial \Omega$ is denoted by $\nu$. For integer $m \geq 0, n\geq
1$, let $H^m(\Omega;\R^n)$ be the standard Sobolev space with a norm
$\|\cdot\|_m$ given by 
$$ 
\|\vec{v}\|_m^2 := \sum_{i=1}^n \sum_{|\alpha|\leq m} \|D^\alpha
v_{x_i}\|_{L^2(\Omega)}^2, \quad \forall \vec{v} \in H^m(\Omega;\R^n).
$$ 
In particular, the norm and inner product of $L^2(\Omega;\R^n) =
H^0(\Omega;\R^n)$ are denoted by $\|\cdot\|_0$ and $(\cdot, \cdot)$,
respectively. For any vector field $\vec{v} \in H^1(\Omega;\R^n)$,
we define
$$ 
\nabla\vec{v} = 
\begin{pmatrix}
\nabla v_1 \\
\vdots \\
\nabla v_n
\end{pmatrix} = 
\begin{pmatrix}
\partial_{x_1} v_{1} & \cdots & \partial_{x_d}v_{1} \\
\vdots & \vdots & \vdots \\
\partial_{x_1}v_{n} & \cdots & \partial_{x_d}v_{n}
\end{pmatrix}
\in \R^{n\times d}.
$$ 
The inner product of the vector is defined as $\vec{v}\cdot \vec{w} =
\sum_{i=1}^n v_i w_i$, for all $\vec{v}, \vec{w}\in \R^n$.  Moreover,
the Frobenious inner product of the matrix is defined as 
$$ 
\langle \bA , \bB \rangle = \bA: \bB = \sum_{i=1}^n \sum_{j=1}^m
a_{ij}b_{ij}, \quad \bA, \bB \in \R^{n\times m}.
$$ 

In this section, we will give a derivation of the models describing
the $N$-phase flows with the effect of pairwise surface tensions. To
this end, we state three assumptions:
\begin{description}
\item[Assumption 1] The $i$-th phase is characterized by $c_i$, which
satisfies $\sum c_i = 1$ and $0\leq c_i \leq 1$. Specifically, $c_i$
corresponds to the volume (or mole) fraction of the $i$-th fluid.
\item[Assumption 2] The free-energy density of the $N$-phase model will
reduce to the corresponding free-energy density of the $L$-phase model
if only $L$ ($2\leq L \leq N-1$) phases are presented. 
\item[Assumption 3] If $N-K$ ($2\leq K \leq N-1$) phases
are not present at the initial time, they will not appear artificially
during the evolution of the system. 
\end{description}

Let $\vec{c} = (c_1, c_2, \cdots, c_N)^T \in \R^N$. Given an
invertible $\bA \in \R^{N\times N}$ and $\vec{b} \in \R^N$, we define
the phase variables $\vec{\phi}$ as 
\begin{equation} \label{equ:phase-variables}
\vec{\phi} = \bA \vec{c} + \vec{b}.
\end{equation}
Let $\sigma_{ij} (1\leq i,j\leq N)$ denote the pairwise surface
tension between phase $i$ and phase $j$ ($\sigma_{ij} = \sigma_{ji}$),
and $\sigma_{ii} = 0$ for $1\leq i \leq N$.  In light of
\cite{boyer2014hierarchy, dong2014efficient, dong2015physical}, we
introduce the free-energy density of the $N$-phase system as  
\begin{equation} \label{equ:energy}
W(\vec{\phi}, \nabla\vec{\phi}) 
:= \sum_{i,j = 1}^N \frac{\eta\lambda_{ij}}{2} \nabla \phi_i \cdot \nabla
\phi_j + \frac{1}{\eta} F(\vec{c})  
= \frac{\eta}{2} ( \bLambda \nabla \vec{\phi} )
: \nabla \vec{\phi} + \frac{1}{\eta} F(\vec{c}), 
\end{equation} 
where the form of nonlinear potential $F(\cdot)$ satisfies the {\bf
  Assumption 2} \cite{boyer2014hierarchy, dong2014efficient,
dong2015physical}, especially when $L=2$:  
\begin{equation} \label{equ:bulk-energy-M2}
F(\vec{c}) = 2\sigma_{ij}[f(c_i) + f(c_j)], \quad 
\mbox{if}~c_i+c_j=1, c_k = 0~(k\neq i,j).
\end{equation}
Here, $f(c) = c^2(1-c^2)$, and the symmetric coefficient matrix
$\bLambda = (\lambda_{ij}) \in \R^{N\times N} $ is assumed to be
constant. We note that the introduction of $\bLambda$ constitutes the
major difference between \eqref{equ:energy} and the $N$-phase models
presented in the literature \cite{boyer2006study, boyer2011numerical,
kim2005phase, kim2009generalized, kim2012phase, lee2008second,
lee2012efficient, lee2012practically, graser2014nonsmooth}. The
constant $\eta > 0$ denotes a characteristic scale of the
interfacial thickness.  With the free-energy density, the
corresponding Liapunov free-energy functional is 
\begin{equation} \label{equ:total-energy}
\begin{aligned}
E(\vec{\phi}) = \int_\Omega W(\vec{\phi}, \nabla \vec{\phi}) \rd x
= \int_\Omega \frac{\eta}{2} ( \bLambda \nabla\vec{\phi} ) :
  \nabla\vec{\phi} + \frac{1}{\eta} F(\vec{\phi}).
\end{aligned}
\end{equation}

Now, we will use \textbf{Assumption 2} to build the relationship
between $\bLambda$, $\bA$, and the pairwise surface tensions. For the
two-phase case, the phase variable $\phi$ satisfies 
$$
c_1 = \frac{1+\phi}{2}, \quad c_2 = \frac{1-\phi}{2},
$$ 
and the free-energy density in \cite{yue2004diffuse} can be written as  
\begin{equation} \label{equ:2-phase-energy}
\begin{aligned}
W(\phi, \nabla \phi) &= \frac{\bar{\lambda}}{2} \nabla\phi \cdot
\nabla\phi + \frac{\bar{\lambda}}{4\epsilon^2} (1-\phi^2)^2 \\
&= \frac{\bar{\lambda}}{2} \nabla (c_1 - c_2) \cdot \nabla(c_1 - c_2) +
\frac{2\sigma_{12}}{\eta} [c_1^2(1-c_1)^2 + c_2^2(1-c_2)^2] \\
&= 2\bar{\lambda} \nabla c_1 \cdot \nabla c_1 + 
\frac{2\sigma_{12}}{\eta} [c_1^2(1-c_1)^2 + c_2^2(1-c_2)^2], 
\end{aligned}
\end{equation}
where $\epsilon = \sqrt{\eta\bar{\lambda}/\sigma_{12}}$.
Moreover, based on the equilibrium 1D surface energy
\cite{yue2004diffuse}, the relationship between $\bar{\lambda}$ and the
interfacial surface tension $\sigma_{12}$ can be derived as 
$$ 
\sigma_{12} = \frac{2\sqrt{2}}{3} \frac{\bar{\lambda}}{\epsilon} =
\frac{2\sqrt{2}}{3} \frac{\epsilon}{\eta}\sigma_{12},
$$ 
which yields 
\begin{equation} \label{equ:2-lambda-tension}
\epsilon = \frac{3}{2\sqrt{2}}\eta, \quad \mbox{and} \quad
\bar{\lambda} = \frac{9}{8}\eta\sigma_{12}.
\end{equation} 
It can be proved in \cite{yue2004diffuse} that
\eqref{equ:2-lambda-tension} gives the interfacial tension in the
sharp-interface limit. Notice that $\sigma = \cO(1)$ such that we have
$\bar{\lambda} = \cO(\eta)$ and $\epsilon = \cO(\eta)$, which is
consistent with the physical model in the two-phase case. We also note
that different phase interfaces have the same interface thickness in
this model. 

Now we assume that only two phases, i.e. $k$ and $l$, are present in
the $N$-phase model, 
$$ 
c_k + c_l = 1, \quad \text{and} \quad c_i = 0~\text{for}~i\neq k, l.
$$ 
Then, the free-energy density \eqref{equ:energy} is shown to be 
$$ 
\begin{aligned}
W(\vec{\phi}, \nabla \vec{\phi}) &= \sum_{i,j = 1}^N
\frac{\eta \lambda_{ij}}{2}(a_{ik} \nabla c_k + a_{il} \nabla c_l)\cdot
(a_{jk} \nabla c_k + a_{jl} \nabla c_l) + \frac{2}{\eta}
[c_k^2(1-c_k)^2 + c_l^2(1-c_l)^2] \\
&= \left( \sum_{i,j=1}^N \frac{\eta \lambda_{ij}}{2} (a_{ik}-a_{il})(a_{jk}
   - a_{jl}) \right) \nabla c_k \cdot \nabla c_k + \frac{2}{\eta}
[c_k^2(1-c_k)^2 + c_l^2(1-c_l)^2]. 
\end{aligned}
$$ 
By comparing the above equation with \eqref{equ:2-phase-energy}, we
immediately have 
\begin{equation} \label{equ:coef-system}
\sum_{i,j=1}^N \lambda_{ij}(a_{ik}-a_{il})(a_{jk}-a_{jl}) =
\frac{4\bar{\lambda}}{\eta} = \frac{9}{2}\sigma_{kl},
\quad 1\leq k < l \leq N,
\end{equation}
where $\sigma_{kl}$ is the interfacial surface tension between phases
$k$ and $l$. Denote $\vec{a}_k = (a_{1k}, a_{2k}, \cdots, a_{Nk})^T
\in \R^N$. Then, \eqref{equ:coef-system} is shown to be the following
matrix equation for $\bLambda$:
\begin{equation} \label{equ:coef-equations} 
(\vec{a}_k - \vec{a}_l)^T \bLambda (\vec{a}_k - \vec{a}_l) =
\frac{9}{2} \sigma_{kl}, \quad 1\leq k < l \leq N.
\end{equation}
Define $\vec{L}_{kl} = \vec{a}_k - \vec{a}_l$. Then,
\eqref{equ:coef-equations} can be recast as  
\begin{equation} \label{equ:coef-equations-matrix} 
(\vec{L}_{kl} \vec{L}_{kl}^T ) : \bLambda = \frac{9}{2}\sigma_{kl},
\quad 1\leq k < l \leq N.
\end{equation} 

\begin{remark}
In most of the well-studied $N$-phase models
\cite{kornhuber2006robust, lee2012efficient, lee2008second,
vanherpe2010multigrid, lee2012practically, graser2014nonsmooth}, the
original concentrations $\vec{c}$ are used as the phase variables and
$\bLambda$ is set as $\sigma \bI$, which means that the models
describe the homogeneous pairwise surface tensions such that
$\sigma_{ij} = \frac{4}{9}\sigma$. 
\end{remark}

\subsection{Solvability and properties of the coefficient matrix}
\label{subsec:SPD} 
Given the phase variables \eqref{equ:phase-variables} and the
free-energy density \eqref{equ:energy}, one basic problem is the
solvability of mixing energy density coefficient matrix $\bLambda =
(\lambda_{ij})$. We note that the number of equations in
\eqref{equ:coef-equations} is $\frac{N(N-1)}{2}$, whereas the number
of unknowns is $\frac{(N+1)N}{2}$. Therefore, we can show only that
$\bLambda$ is unisolvent on the $(N-1)$-dimensional hyperplane.

Notice that $\sum_{i=1}^N c_i = 1$ from \textbf{Assumption 1}. Then, 
$$ 
1 = \vec{1}^T\vec{c} = \vec{1}^T \bA^{-1}(\vec{\phi} -
    \vec{b}) = \vec{d}^T(\vec{\phi} - \vec{b}),
$$ 
where $\vec{1} = (1,1, \cdots, 1)^T \in \R^N$ and $\vec{d} =
\bA^{-T}\vec{1}$. It is easy to check that $\vec{d} \neq 0$,
as $\bA$ is invertible. Then, $\vec{\phi}$ lies in the following
$(N-1)$-dimensional manifold (hyperplane):
\begin{equation} \label{equ:manifold-phi}
\Sigma = \{\vec{\phi} \in \R^N~|~ \vec{d}^T \vec{\phi} - \vec{d}^T
\vec{b} = 1\},
\end{equation}
the tangent space of which is denoted by  
\begin{equation} \label{equ:tangent-space}
T\Sigma = \{\vec{v} \in \R^N~|~ \vec{d}^T \vec{v} = 0\}.
\end{equation}
Furthermore, we have 
$$
\vec{d}^T \vec{a}_{k} = \vec{1}^T \bA^{-1}\vec{a}_k = 1,
\quad \vec{d}^T \vec{L}_{kl} = 0,
$$
which means that $\vec{d} \bot \vec{L}_{kl}$, therefore, $\vec{L}_{kl}
\in T\Sigma$. Let $\bP = \bI - \frac{\vec{d}\vec{d}^T}{|\vec{d}|^2}$
be the orthogonal projection to $T\Sigma$ such that $\bP^T = \bP$ and 
$$ 
\bP \vec{L}_{kl} = \vec{L}_{kl}, \quad \bP \vec{L}_{kl}
\vec{L}_{kl}^T \bP =  \vec{L}_{kl} \vec{L}_{kl}^T. 
$$ 
Then, \eqref{equ:coef-equations-matrix} is shown to be 
\begin{equation} \label{equ:coef2-equations-matrix}
(\vec{L}_{kl} \vec{L}_{kl}^T ) : \tilde{\bLambda} =
\frac{9}{2}\sigma_{kl}, \quad 1\leq k < l \leq N, 
\end{equation}
where $\tilde{\bLambda} = \bP \bLambda \bP$. Based on the property of
the symmetric matrix space, we will establish the unique solvablity of
$\tilde{\bLambda}$ next.

\begin{definition}
For a given $n$-dimensional vector space $U \subset \R^N$, $\bP: \R^N
\mapsto U$ is the orthogonal projection. Define the symmetric matrix
space on $U$ as 
$$
\mS(U) = \{\bP\bM\bP~|~\bM\in \R^{N\times N}, \bM^T = \bM\}.
$$ 
\end{definition}
As the dimension of $\ker(\bP)$ is $N-n$, then $\dim(\mS(U)) =
\frac{n(n+1)}{2}$. Let inner product $\langle \cdot, \cdot\rangle$ be
the Frobenious inner product, then it is easy to determine that
$(\mS(U), \langle\cdot, \cdot\rangle)$ is a Hilbert space. 

In Lemma 2.2 of \cite{hu2015finite}, Hu determined a crucial
relationship between the $n$-dimensional simplex and the
$n$-dimensional symmetric matrix space. We extend this lemma to a
hyperplane to obtain the following lemma. 

\begin{lemma} \label{lem:simplex-symmetric}
$\{\vec{L}_{kl}\vec{L}_{kl}^T, 1\leq k<l\leq N\}$ forms a basis
of $\mS(T\Sigma)$.  
\end{lemma}
\begin{proof}
It is easy to see that $\vec{L}_{kl}\vec{L}_{kl}^T$ constitutes a
symmetric matrix of rank one in $\mS(T\Sigma)$ and that  
$$ 
\dim\left( \{\vec{L}_{kl}\vec{L}_{kl}^T, 1\leq k < l \leq
N\}\right) \leq \frac{N(N-1)}{2} = \dim(\mS(T\Sigma)).
$$

On the other hand, notice that $\bA$ is invertible. Then, $[\vec{a}_1,
\cdots, \vec{a}_N]$ comprise an $(N-1)$-dimensional simplex on the
hyperplane 
$$ 
\cM = \{\vec{v} \in \R^N~|~ \vec{d}^T \vec{v} = 1\}.
$$ 
If $\sum_{kl} \alpha_{kl} \vec{L}_{kl} \vec{L}_{kl}^T = 0$, then by
testing the normal vector $\vec{n}_N \in \cM$ of the
$(N-2)$-dimensional hyperplane $[\vec{a}_1, \cdots, \vec{a}_{N-1}]$ on 
both sides (see Figure \ref{fig:schematic-fig1}), we obtain  
\[
\sum_{k=1}^{N-1} \alpha_{kN} \vec{L}_{kN} \vec{L}_{kN}^T
\vec{n}_N = 0,
\]
as
\[
\vec{L}_{kl}^T \vec{n}_N = 0, \quad 1 \leq k < l \leq N-1.
\]
Notice that $\vec{L}_{kN}^T \vec{n}_N \neq 0$ and $\{\vec{L}_{kN},
1\leq k \leq N-1\}$ are linear independent. Then, we immediately have  
\[
\alpha_{kN} = 0, \quad 1\leq k \leq N-1.
\]
By a similar argument, we can prove that $\alpha_{kl} = 0$, which
means that $\{\vec{L}_{kl}\vec{L}_{kl}^T, 1\leq k < l \leq N\}$ forms
a basis of $\mS(T\Sigma)$. 
\end{proof}

\begin{theorem} \label{thm:uni-solvent}
Assume that $\bLambda$ satisfies the linear algebraic system
\eqref{equ:coef-equations}. Then, $\tilde{\bLambda} = \bP\bLambda\bP$
is uniquely determined by the interfacial surface tension
$\sigma_{kl}$. 
\end{theorem}
\begin{proof} 
This theorem can be directly proved by Lemma
\ref{lem:simplex-symmetric} and the Riesz representation theorem in
Hilbert space.
\end{proof}

\begin{figure}[!htbp]
\centering 
\captionsetup{justification=centering}
\subfloat[Schematic diagram for solvability, $N=3$]{
  \includegraphics[width=0.45\textwidth]{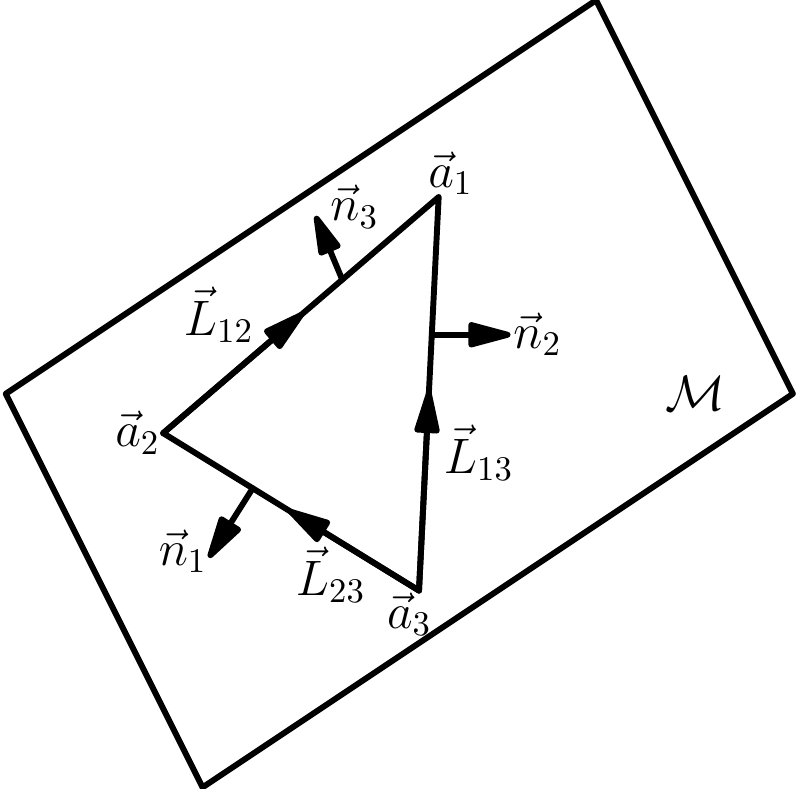} 
  \label{fig:schematic-fig1}
}%
\subfloat[Schematic diagram for SPD property, $N=4$]{
  \includegraphics[width=0.45\textwidth]{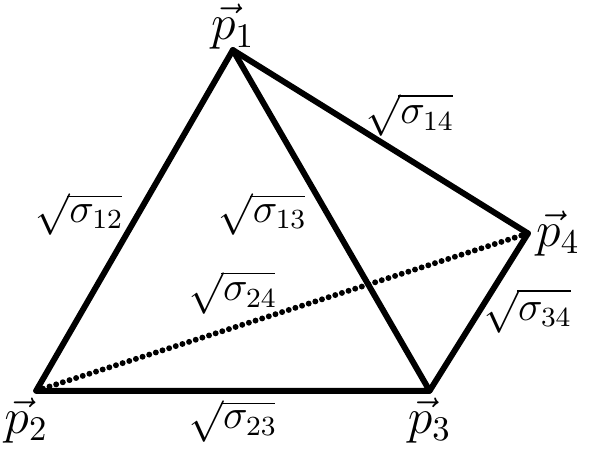} 
  \label{fig:schematic-fig2}
}
\caption{Schematic diagrams for solvability and SPD property}
\label{fig:systemic-figs}
\end{figure}

In Theorem \ref{thm:uni-solvent}, we build a bridge between the
$N$-phase models and the vertices of the $(N-1)$-dimensional simplex.
Based on this idea, we will present a sufficient and necessary
condition for $\tilde{\bLambda}$ to be symmetric positive-definite
(SPD) on the tangent space $T\Sigma$.

\begin{theorem} \label{thm:SPD}
The following statements are equivalent: 
\begin{enumerate}
\item $\tilde{\bLambda}$ is SPD on the tangent space $T\Sigma$. 
\item For any (or there exists) $1\leq m \leq N$, the matrix
$\tilde{\bsigma}^m = (\tilde{\sigma}_{ij}^m) \in \R^{(N-1)\times(N-1)}$
is SPD, where $\tilde{\bsigma}^m$ is obtained from $\tilde{\bsigma}$
by removing the $m$-th row and column: 
\begin{equation} \label{equ:SPD-cond1}
(\tilde{\bsigma})_{ij} = 
\frac{\sigma_{im} + \sigma_{jm} - \sigma_{ij}}{2}, \quad 1\leq i, j
\leq N.
\end{equation}
\item The surface tensions can compose a non-degenerate
$(N-1)$-dimensional simplex $K = [\vec{p}_1, \cdots, \vec{p}_N]$ with
$|\vec{p}_i - \vec{p}_j| = \sqrt{\sigma_{ij}}$.
\end{enumerate}
\end{theorem}

\begin{proof}
{\textbf{$1\Leftrightarrow 2$}}:
It is easy to check that $\{\vec{L}_{km}, 1\leq k \leq N, k\neq m\}$
forms a basis of $T\Sigma$. Then, any $\vec{v} \in T\Sigma$ can be
written as
\[
\vec{v} = \sum_{k\neq m} \alpha_k \vec{L}_{km} = [\vec{L}_{1m},
  \cdots, \vec{L}_{(k-1)m}, \vec{L}_{(k+1)m}, \cdots, \vec{L}_{Nm}]
  \vec{\alpha}^m := \boldsymbol{L}^m \vec{\alpha}^m, 
\]
where $\vec{\alpha}^m = (\alpha_1, \cdots, \alpha_{m-1}, \alpha_{m+1},
\cdots, \alpha_N)^T$ and $\boldsymbol{L}^m \in \R^{N\times(N-1)}$.
It can easily be seen from \eqref{equ:coef2-equations-matrix} that  
$$ 
\frac{9}{2} \sigma_{ij} = \vec{L}_{ij}^T \tilde{\bLambda}
\vec{L}_{ij} = (\vec{L}_{im} - \vec{L}_{jm})^T \tilde{\bLambda}
(\vec{L}_{im} - \vec{L}_{jm}) = \frac{9}{2}\sigma_{im} +
\frac{9}{2}\sigma_{jm} - 2\vec{L}_{im}^T \tilde{\bLambda}
\vec{L}_{jm}.   
$$ 
Then, we have 
\[
\begin{aligned}
\vec{v}^T \tilde{\bLambda} \vec{v} = (\vec{\alpha}^m)^T
[(\boldsymbol{L}^m)^T \tilde{\bLambda} \boldsymbol{L}^m]
\vec{\alpha}^m = \frac{9}{2} (\vec{\alpha}^m)^T \tilde{\bsigma}^m
\vec{\alpha}^m, 
\end{aligned}
\]
which means that the SPD of $\tilde{\bLambda}$ on $T\Sigma$ is equivalent
to the SPD of $\tilde{\bsigma}^m$ defined in \eqref{equ:SPD-cond1}. 

{\textbf{$2\Rightarrow 3$}}: We choose $m=N$ for simplicity. For the
SPD matrix $\tilde{\bsigma}^N \in \R^{(N-1)\times(N-1)}$, there exists
an invertible matrix $\bT$ such that $\tilde{\bsigma}^N = \bT^T\bT$.
Define $\vec{p}_N = \vec{0}$ and 
$$ 
\vec{p}_i = (t_{1i}, t_{2i}, \cdots, t_{(N-1)i})^T\in \R^{N-1}, \quad
i = 1, \cdots, N-1.
$$
Then, we immediately know that $[\vec{p}_1, \cdots, \vec{p}_N]$ form a
non-degenerate simplex. By checking the diagonal terms of
$\bT^T\bT$, we have $|\vec{p}_i| = \sqrt{\sigma_{iN}}$.
Furthermore, the off-diagonal terms of $\bT^T\bT$ imply that  
$$ 
\frac{\sigma_{iN} + \sigma_{jN} - \sigma_{ij}}{2} = 
\vec{p}_i \cdot \vec{p}_j = 
\frac{|\vec{p}_i|^2 + |\vec{p}_j|^2 - |\vec{p}_i - \vec{p}_j|^2}{2} =
\frac{\sigma_{iN} + \sigma_{jN} - |\vec{p}_i -
\vec{p}_j|^2}{2}, 
$$ 
which yields $|\vec{p}_i-\vec{p}_j| = \sqrt{\sigma_{ij}}$.

{\textbf{$3 \Rightarrow 2$}}:  If 3 holds, then there exists an affine
mapping from unit simplex $\hat{K} = [\vec{e}_1, \cdots,
\vec{e}_{N-1}, \vec{0}]$ to $K = [\vec{p}_1, \cdots, \vec{p}_N]$:
\[
\vec{y} = G(\hat{y}) = \bT \hat{y} + \vec{p}_N, \quad \forall \hat{y}
\in \hat{K}. 
\]
Then, it is easy to check that 
\[
\bT = (\vec{p}_1-\vec{p}_N, \vec{p}_2-\vec{p}_N, \cdots,
    \vec{p}_{N-1}-\vec{p}_N),
\]
which implies that  
$$ 
\tilde{\bsigma}^N = \bT^T \bT.
$$
The non-degenerate property of $K$ means that $\det(\bT) \neq 0$,
which leads to the SPD of $\tilde{\bsigma}^N$. 
\end{proof}

\begin{remark} 
Statement 3 in Theorem \ref{thm:SPD} is the geometric condition
(we call it {\it simplicial condition}) for the pairwise surface
tensions, see Figure \ref{fig:schematic-fig2} for the case in which
$N=4$. We can easily find that a necessary condition for the SPD
property of $\tilde{\bLambda}$ on the tangent space is  
\begin{equation} \label{equ:triangle-surface-tension}
|\sqrt{\sigma_{ij}} - \sqrt{\sigma_{jk}}| < \sqrt{\sigma_{ik}} <
\sqrt{\sigma_{ij}} + \sqrt{\sigma_{jk}},
\quad \mbox{for different }i,j,k.
\end{equation}
For the case in which $N=3$, condition
\eqref{equ:triangle-surface-tension} is obviously the  sufficient
condition from Theorem \ref{thm:SPD}. However, for the case in which
$N \geq 4$, \eqref{equ:triangle-surface-tension} is not sufficient,
which makes it difficult to extend the existing three-phase models
\cite{kim2004conservative, kim2005phase, boyer2006study, kim2007phase,
boyer2011numerical} to an arbitrary number of phases. 
\end{remark}

\subsection{$N$-phase Allen-Cahn equations}
It is well known that $N$-phase Allen-Cahn equations can be derived as
the gradient flow, which implies that 
\begin{equation} \label{equ:gradient-flow} 
\gamma \pd{\vec{\phi}}{t} = -\grad W(\vec{\phi}, \nabla \vec{\phi}).
\end{equation} 
Here, parameter $\gamma$ is set as $\cO(\eta)$ consistent with the
mean curvature flow for the two-phase case \cite{evans1992phase}.
In light of \eqref{equ:grad} below, we know that $\grad W$ belongs to
the dual space of $T\Sigma$. Therefore, the left-hand side of the
gradient flow \eqref{equ:gradient-flow} should also be interpreted as
the dual space of $T\Sigma$, which means that the metric on $T\Sigma$
must be considered. First, we define the Sobolev spaces on manifold
$\Sigma$ and tangent space $T\Sigma$ as 
\begin{equation} \label{equ:spaces} 
\begin{aligned}
H^1(\Sigma) &:= H^1(\Omega;\R^n) \cap \Sigma = \{\vec{\phi} \in
H^1(\Omega;\R^n)~|~\vec{d}^T \vec{\phi} - \vec{d}^T\vec{b} = 1\},\\ 
H^1(T\Sigma) &:= H^1(\Omega;\R^n) \cap T\Sigma = \{\vec{v}\in
H^1(\Omega;\R^n)~|~\vec{d}^T \vec{v} = 0\}. 
\end{aligned}
\end{equation}
It can be seen that $H^1(\Sigma) = H^1(T\Sigma) + \vec{b} +
\frac{\vec{d}}{|\vec{d}|^2}$.

Let $\vec{\phi} \in H^1(\Sigma)$ and $(\nabla\vec{\phi}) \nu = 0$ on
$\partial\Omega$.  For any $\vec{v}\in H^1(T\Sigma)$, we obtain the
gradient of $E(\vec{\phi}, \nabla \vec{\phi})$ on the manifold
$\Sigma$ as 
\begin{equation} \label{equ:grad} 
\begin{aligned}
\langle \grad E, \vec{v} \rangle &= \int_{\Omega} \left. \od{}{\theta}
W(\vec{\phi} + \theta\vec{v}, \nabla\vec{\phi} + \theta \nabla\vec{v})
  \right|_{\theta = 0} \rd x \\
&= \int_{\Omega} \eta (\bLambda \nabla\vec{\phi}): \nabla \vec{v} +
\frac{1}{\eta} (\bA^{-T}\pd{F}{\vec{c}}) \cdot \vec{v} ~\rd x \\
&= \int_{\Omega}\eta (\bP \bLambda \bP \nabla \vec{\phi}): \nabla\vec{v} +
\frac{1}{\eta} (\bP\bA^{-T} \pd{F}{\vec{c}}) \cdot \vec{v} ~\rd x.
\end{aligned}
\end{equation}

Denote the manifold (hyperplane) of the concentration as 
$$ 
\Sigma_c = \{\vec{c}\in \R^N ~|~ \vec{1}^T \vec{c} = 1\}.
$$ 
Then, we have 
$$ 
\Sigma_c \xlongrightarrow{\bA\vec{c} + \vec{b}} \Sigma, \qquad T\Sigma_c
\xlongrightarrow{\bA} T\Sigma.
$$ 
If a given inner product $(\cdot, \cdot)_{\bX_c}$ is used for
$T\Sigma_c$, then the induced inner product for $T\Sigma$ will be  
\begin{equation} \label{equ:induced-inner}
(\cdot, \cdot)_{\bX} := (\bX_c\bA^{-1}\cdot,
\bA^{-1}\cdot)_{l^2},
\end{equation}
where $\bX = \bA^{-T}\bX_c\bA^{-1}$. When choosing $\bX_c = \bM_c^{\rm
AC}$, the weak formulation of the $N$-phase Allen-Cahn equations can
be written as  
\begin{equation} \label{equ:N-AC-weak}
\gamma \int_\Omega (\bM_c^{{\rm AC}} \bA^{-1}\pd{\vec{\phi}}{t}) \cdot
(\bA^{-1}\vec{v}) ~\rd x + \int_\Omega \eta (\tilde{\bLambda}
\nabla\vec{\phi}) : \nabla \vec{v} + \frac{1}{\eta} (\bP\bA^{-T}
\pd{F}{\vec{c}}) \cdot \vec{v}~\rd x = 0,
\quad \forall \vec{v} \in H^1(T\Sigma),
\end{equation}
whereas the strong form can be written as  
\begin{equation} \label{equ:N-AC-PDE}
\left\{
\begin{aligned}
\gamma\bP \bA^{-T} \bM_c^{{\rm AC}} \bA^{-1} \pd{\vec{\phi}}{t} -\nabla \cdot
(\eta \tilde{\bLambda} \nabla\vec{\phi}) + \frac{1}{\eta} \bP\bA^{-T}
\pd{F}{\vec{c}} &= 0, \quad \text{in}~\Omega\times(0,T], \\
(\nabla \vec{\phi}) \nu & = 0, \quad \text{on}~\partial\Omega
\times(0,T].
\end{aligned}
\right.
\end{equation}

We will prove that the $N$-phase Allen-Cahn equations
\eqref{equ:N-AC-weak} are independent of the choice of $\bA$ in
the following theorem.

\begin{theorem} \label{thm:NAC-invarient}
Let $\bP_c = \bI - \frac{\vec{1}\vec{1}^T}{\vec{1}^T\vec{1}}$,
$\tilde{\bLambda}_c = \bP_c \bLambda_c \bP_c$, and $\vec{v}_c =
\bA^{-1}\vec{v}$. Then, \eqref{equ:N-AC-weak} is equivalent to 
\begin{equation} \label{equ:N-AC-weak-c} 
\gamma \int_\Omega (\bM_c^{\rm AC} \pd{\vec{c}}{t}) \cdot \vec{v}_c ~\rd
x = \int_\Omega \eta(\tilde{\bLambda}_c \nabla\vec{c}): \nabla
\vec{v}_c + \frac{1}{\eta}(\bP_c \pd{F}{\vec{c}})\cdot \vec{v}_c~\rd
x, \quad \forall \vec{v}_c \in H^1(T\Sigma_c),
\end{equation} 
or to the strong form 
\begin{equation} \label{equ:N-AC-PDE-c}
\left\{
\begin{aligned}
\gamma \bP_c \bM_c^{\rm AC} \pd{\vec{c}}{t} -\nabla \cdot
(\eta \tilde{\bLambda}_c
\nabla\vec{c}) + \frac{1}{\eta} \bP_c \pd{F}{\vec{c}} &= 0,
\quad \mbox{in}~\Omega\times(0,T], \\
(\nabla \vec{c}) \nu & = 0, \quad \mbox{on}~\partial\Omega
\times(0,T].
\end{aligned}
\right.
\end{equation}
\end{theorem}
\begin{proof}
It is easy to check that $\vec{v}_c\in H^1(T\Sigma_c)$.  Then,
\eqref{equ:N-AC-weak} is shown to be 
\begin{equation} \label{equ:N-AC-phi2c} 
\begin{aligned}
\gamma \int_\Omega (\bM_c^{\rm AC}\bA^{-1}\pd{\vec{\phi}}{t}) \cdot
\vec{v}_c ~\rd x &= \int_\Omega \eta (\bA^T \tilde{\bLambda}
\nabla\vec{\phi}) : \nabla \vec{v}_c + \frac{1}{\eta} (\bA^T\bP\bA^{-T}
\pd{F}{\vec{c}}) \cdot \vec{v}_c ~\rd x
\\
& = 
\int_\Omega \eta(\bP_c \bA^T \tilde{\bLambda} \nabla\vec{\phi}): \nabla
\vec{v}_c + \frac{1}{\eta} (\bP_c\bA^T\bP\bA^{-T} \pd{F}{\vec{c}})
\cdot \vec{v}_c ~\rd x \\
& = 
\int_\Omega \eta [(\bA\bP_c)^T \tilde{\bLambda} (\bA\bP_c)
\nabla\vec{c}]
: \nabla \vec{v}_c + \frac{1}{\eta} (\bP_c\bA^T\bP\bA^{-T}
\pd{F}{\vec{c}}) \cdot \vec{v}_c ~\rd x. 
\end{aligned}
\end{equation}
For the right-hand side of \eqref{equ:N-AC-phi2c}, it is easy to
determine that 
\begin{equation} \label{equ:AC-invarient1} 
\bP_c \bA^{T}\bP\bA^{-T} = \bP_c (\bI -
\frac{\vec{1}\vec{1}^T\bA^{-1}\bA^{-T}}{|\vec{d}|^2}) = \bP_c.
\end{equation} 
On the other hand, we have 
$$
\bA\bP_c(\vec{e}_k - \vec{e}_l) = \bA(\vec{e}_k - \vec{e}_l) =
\vec{a}_k - \vec{a}_l = \vec{L}_{kl}.
$$ 
Thus, when  
$$ 
[(\vec{e}_k - \vec{e}_l)(\vec{e}_k - \vec{e}_l)^T] : [(\bA\bP_c)^T
\tilde{\bLambda} (\bA\bP_c)] = (\vec{L}_{kl}\vec{L}_{kl}^T) :
\tilde{\bLambda} = \frac{9}{2}\sigma_{kl}, 
$$ 
is taken together with the unisolvent property in Theorem
\ref{thm:uni-solvent}, we obtain
\begin{equation} \label{equ:AC-invarient2} 
(\bA\bP_c)^T \tilde{\bLambda} (\bA\bP_c) =  \tilde{\bLambda}_c.
\end{equation} 
Take \eqref{equ:AC-invarient1} and \eqref{equ:AC-invarient2} into
\eqref{equ:N-AC-phi2c} to obtain the desired results. 
\end{proof}

Since $\bM_c^{\rm AC}$ is SPD on $T\Sigma_c$, by taking $\vec{v} =
\vec{\phi}_t$ in \eqref{equ:N-AC-weak}, we immediately find the
following energy law for $N$-phase Allen-Cahn equations:
\begin{equation} \label{equ:N-AC-energy} 
\od{E(\vec{\phi})}{t} = - \gamma \int_{\Omega}
(\bM_c^{\rm AC} \bA^{-1} \vec{\phi}_t) \cdot (\bA^{-1}\vec{\phi}_t)
~\rd x  = -\gamma \int_{\Omega} (\bM_c^{\rm AC} \vec{c}_t) \cdot
  \vec{c}_t~\rd x \leq 0,
\end{equation} 
which depends only on the dynamics of concentrations, as expected.

\begin{remark}
If $\pd{F}{\vec{c}} = (f'(c_1), f'(c_2), \cdots, f'(c_N))^T$, then
\eqref{equ:N-AC-PDE-c} implies that 
$$ 
\bP_c \pd{F}{\vec{c}} = (\bI - \frac{\vec{1}\vec{1}^T}{N})
  \pd{F}{\vec{c}} = \pd{F}{\vec{c}} - \frac{\vec{1}}{N}\sum_{i=1}^N
  f'(c_i) = \pd{F}{\vec{c}} + \beta(\vec{c})\vec{1}, 
$$ 
where $\beta(\vec{c}) = -\frac{1}{N}\sum_{i=1}^N f'(c_i)$ is exactly
the variable Lagrangian multiplier used in the existing works \cite{
kornhuber2006robust, lee2012efficient, lee2008second,
vanherpe2010multigrid, lee2012practically, graser2014nonsmooth}.
\end{remark}

\subsection{$N$-phase Cahn-Hilliard equations}

For Cahn-Hilliard equations, it is well known that the Hele-Shaw flow
constitutes limiting dynamics in the two-phase case
\cite{chen1991uniqueness, alikakos1994convergence}. Let $\bM_c^{\rm
CH}$ be the mobilities associated with $\vec{c}$ that is SPD on
$T\Sigma_c$. Similar to the argument for the $N$-phase Allen-Cahn
equations, by choosing $(\cdot,\cdot)_{\bX_c} = (\cdot,\cdot)_{\bI_c}$
as the inner product on $T\Sigma_c$, the $N$-phase Cahn-Hilliard
equations under the induced inner product \eqref{equ:induced-inner}
are 
\begin{equation} \label{equ:N-CH-weak}
\left\{
\begin{aligned}
\int_{\Omega} (\bA^{-1} \pd{\vec{\phi}}{t}) \cdot (\bA^{-1}\vec{q}) ~\rd
x &= -\int_\Omega (\bM_c^{\rm CH} \bA^{-1}\nabla \vec{w}) :
(\bA^{-1}\nabla \vec{q}) ~\rd x, \quad \forall \vec{q} \in
H^1(T\Sigma), \\
\int_\Omega (\bA^{-1}\vec{w}) \cdot (\bA^{-1}\vec{v})~\rd x &=
\int_\Omega \eta (\tilde{\bLambda} \nabla \vec{\phi}) : \nabla
\vec{v} + \frac{1}{\eta}(\bP\bA^{-T}\pd{F}{\vec{c}}) \cdot
\vec{v} ~\rd x, \quad \forall \vec{v} \in H^1(T\Sigma), 
\end{aligned}
\right.
\end{equation}
where $\vec{w}$ denotes the chemical potentials. In light of the weak
formulation \eqref{equ:N-CH-weak}, the strong form of $N$-phase
Cahn-Hilliard equations can be written as  
\begin{equation} \label{equ:N-CH-PDE}
\left\{
\begin{aligned}
\bP\bA^{-T}\bA^{-1} \pd{\vec{\phi}}{t} &= \nabla \cdot [
(\bA^{-1}\bP)^T\bM_c^{\rm CH}
    \bA^{-1}\bP \nabla \vec{w}], \quad
\text{in}~\Omega\times(0,T], \\
\bP\bA^{-T}\bA^{-1}\vec{w} &= - \nabla \cdot
(\eta \tilde{\bLambda}\nabla \vec{\phi}) +
\frac{1}{\eta}\bP\bA^{-1}\pd{F}{\vec{c}}, \quad
\text{in}~\Omega \times(0,T], \\
(\nabla\vec{\phi})\nu = (\nabla\vec{w})\nu &= 0, \quad
\text{on}~\partial\Omega\times(0,T].
\end{aligned}
\right.
\end{equation}

Similar to Theorem \ref{thm:NAC-invarient}, we have the following
theorem for the invariant dynamics of concentrations for $N$-phase
Cahn-Hilliard equations.

\begin{theorem} \label{thm:NAC-invarient}
Let $\vec{v}_c = \bA^{-1}\vec{v}$ and $\vec{w}_c = \bA^{-1}\vec{w}$.
Then, \eqref{equ:N-CH-weak} is equivalent to 
\begin{equation} \label{equ:N-CH-weak-cw}
\left\{
\begin{aligned}
\int_{\Omega} \pd{\vec{c}}{t} \cdot \vec{q}_c ~\rd x &= -\int_\Omega
(\bM_c^{\rm CH} \nabla \vec{w}_c) : \nabla \vec{q}_c ~\rd x, \quad
 \forall \vec{q}_c \in
H^1(T\Sigma_c), \\
\int_\Omega \vec{w}_c \cdot \vec{v}_c~\rd x &= \int_\Omega
(\eta \tilde{\bLambda}_c \nabla \vec{c}) : \nabla \vec{v}_c +
\frac{1}{\eta} (\bP_c \pd{F}{\vec{c}}) \cdot \vec{v}_c ~\rd x, \quad
\forall \vec{v}_c \in H^1(T\Sigma_c), 
\end{aligned}
\right.
\end{equation}
or to the strong form 
\begin{equation} \label{equ:N-CH-PDE-cw}
\left\{
\begin{aligned}
\bP_c\pd{\vec{c}}{t} &= \nabla \cdot ( \bP_c\bM_c^{\rm CH}
    \bP_c \nabla \vec{w}_c), \quad
\text{in}~\Omega\times(0,T], \\
\bP_c \vec{w}_c &= - \nabla \cdot
(\eta \tilde{\bLambda}_c\nabla \vec{c}) +
\frac{1}{\eta}\bP_c \pd{F}{\vec{c}}, \quad
\text{in}~\Omega \times(0,T], \\
(\nabla\vec{c})\nu = (\nabla\vec{w}_c)\nu &= 0, \quad
\text{on}~\partial\Omega\times(0,T].
\end{aligned}
\right.
\end{equation}
\end{theorem}

It is easy to verify the global mass conservation and energy
law of the $N$-phase Cahn-Hilliard model. First, by the first
equation of \eqref{equ:N-CH-weak-cw}, we have  
\begin{equation} \label{equ:N-CH-mass}
\od{}{t}\int_\Omega \bP_c \vec{c} ~\rd x = 0.
\end{equation}
Note that $\vec{1}^T\vec{c} = 0$. Then, we have $\od{}{t}\int_\Omega
\vec{c}~\rd x = 0$. Further, by taking $\vec{v} =
\vec{\phi}_t$ and $\vec{q} = \vec{w}$ in \eqref{equ:N-CH-weak}, we
obtain
\begin{equation} \label{equ:N-CH-momentum}
\begin{aligned}
\od{E(\vec{\phi})}{t} &= \int_\Omega (\bA^{-1}\vec{w})\cdot
(\bA^{-1}\vec{\phi}_t) ~\rd x \\
& = -\int_\Omega (\bM_c^{\rm CH} \bA^{-1}\nabla \vec{w}) \cdot
(\bA^{-1}\nabla \vec{w})  ~\rd x 
= -\int_{\Omega} (\bM_c^{\rm CH} \nabla\vec{w}_c) \cdot
\nabla\vec{w}_c ~\rd x \leq 0. 
\end{aligned}
\end{equation}
From \eqref{equ:N-CH-momentum}, we see that the $N$-phase
Cahn-Hilliard equations describe the energy law in a conservation
system, as for the two-phase case.

\subsection{Determine of $\bM_c^{\rm AC}$ and $\bM_c^{\rm CH}$, and
choices of $F(\cdot)$}
Now, we will use the {\bf Assumption 3}, to determine the $\bM_c^{\rm
AC}$ and $\bM_c^{\rm CH}$ appearing in \eqref{equ:N-AC-PDE-c} and
\eqref{equ:N-CH-PDE-cw}, respectively. First, since SPD operator is
invertible, we know that for any $\bW_c$ SPD on $T\Sigma_c$, there
uniquely exists a $\bW_c^{\dagger}$, such that $\bW_c \bW_c^{\dagger}
= \bW_c^{\dagger}\bW_c = \bI_c$. Clearly, $\bW_c$ is also SPD. By
direct calculation, 
$$ 
[\nabla \cdot (\tilde{\bLambda}_c \nabla \vec{c})]_i = \sum_{k=1}^N
\sum_{j=1}^d \partial_{x_j} (\tilde{\bLambda}_c)_{ik} \partial_{x_j}
c_k = \sum_{k=1}^N (\tilde{\bLambda}_c)_{ik} \Delta c_k =
(\tilde{\bLambda}_c \Delta \vec{c})_i. 
$$ 
Then, \eqref{equ:N-AC-PDE-c} can be recast as 
$$ 
\gamma \frac{\partial \vec{c}}{\partial t} - \eta (\bM_c^{{\rm
AC},\dagger} \tilde{\bLambda}_c) \Delta \vec{c} + \frac{1}{\eta}
(\bM_c^{{\rm AC}, \dagger} \bP_c) \frac{\partial F}{\partial \vec{c}}
= 0.
$$ 
Therefore, the {\bf Assumption 3} is equivalent to the following
property: 
$$ 
\mbox{If~}c_i = 0,\mbox{~then~}-\eta \sum_{j=1}^N (\bM_c^{{\rm
AC},\dagger} \tilde{\bLambda}_c)_{ij} \Delta c_j + \frac{1}{\eta}
\left( \bM_c^{{\rm AC},\dagger} \bP_c \frac{\partial F}{\partial
\vec{c}} \right)_i = 0,
$$ 
which requires that both the nonlinear potential term and the
second-order differential term should vanish identically.
Therefore, it is in particular needed that 
\begin{equation} \label{equ:Mob-differential}
\mbox{If~}c_i = 0,\mbox{~then~} \sum_{j=1}^N (\bM_c^{{\rm
AC},\dagger} \tilde{\bLambda}_c)_{ij} \Delta c_j = 0.
\end{equation}

\begin{lemma} \label{equ:determine-Mob}
For any $N \geq 2$, \eqref{equ:Mob-differential} holds if and only if
there exists a constant $C$ such that 
\begin{equation} \label{equ:Mob-AC-general} 
\bM_c^{{\rm AC},\dagger} \tilde{\bLambda}_c = C \bI_c. 
\end{equation} 
\end{lemma}
\begin{proof}
It is straightforward that $\bM_c^{{\rm AC},\dagger}
\tilde{\bLambda}_c$ is a linear operator from $T\Sigma_c$ to
$T\Sigma_c$. When $N=2$, clearly \eqref{equ:Mob-AC-general} is true as
$\dim(T\Sigma_c) = 1$. When $N\geq 3$, consider the following set of
basis of $T\Sigma_c$:
$$ 
\{\vec{e}_1 - \vec{e}_N, \vec{e}_2 - \vec{e}_N, \cdots, \vec{e}_{N-1}
- \vec{e}_N\}.
$$ 
From the property \eqref{equ:Mob-differential}, each basis forms an
invariant 1-dimensional subspace under $\bM_c^{{\rm AC},\dagger}
\tilde{\bLambda}_c$, namely 
$$ 
\bM_c^{{\rm AC},\dagger} \tilde{\bLambda}_c (\vec{e}_i - \vec{e}_N) =
\beta_i (\vec{e}_i - \vec{e}_N), \qquad i = 1, \cdots, N-1.
$$ 
Note that for any $1 \leq i < j \leq N-1$, $\{\vec{e}_i - \vec{e}_j\}$
is also an invariant 1-dimensional subspace under $\bM_c^{{\rm
AC},\dagger} \tilde{\bLambda}_c$. Hence, 
$$ 
\bM_c^{{\rm AC},\dagger} \tilde{\bLambda}_c (\vec{e}_i - \vec{e}_j) =
\bM_c^{{\rm AC},\dagger} \tilde{\bLambda}_c [(\vec{e}_i - \vec{e}_N) -
(\vec{e}_i - \vec{e}_N)] = \beta_i (\vec{e}_i - \vec{e}_N) - \beta_j
(\vec{e}_j - \vec{e}_N) \in \{\vec{e}_i - \vec{e}_j\},
$$ 
which implies that $\beta_i = \beta_j$. Therefore, there exists a
constant $C$ such that $\beta_i = C$ for all $1\leq i \leq N-1$, which
gives rise to \eqref{equ:Mob-AC-general}.
\end{proof}

For conciseness, the constant $C$ can be absorbed into the parameter
$\gamma$ in the $N$-phase Allen-Cahn equation. Therefore, we choose 
\begin{equation} \label{equ:Mob-AC}
\bM_c^{\rm AC} = \tilde{\bLambda}_c.
\end{equation}
In the similar manner, the {\bf Assumption 3} implies the
following choice of $\bM_c^{\rm CH}$ for the  $N$-phase Cahn-Hilliard
equations  
\begin{equation} \label{equ:Mob-CH} 
\bM_c^{\rm CH}\tilde{\bLambda}_c = M_0 \bI_c, \quad \mbox{or}\quad
\bM_c^{\rm CH} = M_0 \tilde{\bLambda}_c^{\dagger},
\end{equation} 
where the positive constant $M_0$ is called the mobility. 

The construction of the nonlinear potential $F(\cdot)$ satisfying the
{\bf Assumption 3} is very challenging. This problem is entirely 
answered for the simplest case in which the pairwise surface tensions
are homogeneous, namely $\sigma_{ij} = \sigma$, 
\begin{equation} \label{equ:homo-F}
F^{\sigma}(\vec{c}) := F^{\sigma}_0(\vec{c}) + F^{\sigma}_1(\vec{c}),
\end{equation}
where 
$$
F^{\sigma}_0(\vec{c}) = 2\sigma \sum_{i=1}^N f(c_i), \qquad
F^{\sigma}_1(\vec{c}) =  
\begin{cases}
0, & N=2,3, \\
8\sigma \sum\limits_{i_1 < i_2 < i_3 < i_4}
c_{i_1} c_{i_2} c_{i_3} c_{i_4}, & N \geq 4.
\end{cases}
$$
We refer to the Proposition 3.3 in \cite{boyer2014hierarchy}.  For the
inhomogeneous case, we consider the following nonlinear potential in
this paper,
\begin{equation} \label{equ:nonlinear-potential} 
F^{\sigma_{ij}}(\vec{c}) := F_0^{\sigma_{ij}}(\vec{c}) + s
F_1^{\sigma_{ij}}(\vec{c}) 
\end{equation}
where 
$$
F_0^{\sigma_{ij}}(\vec{c}) = \sum_{i,j=1}^N \sigma_{ij}[f(c_i) +
f(c_j) - f(c_i+c_j)], \qquad F_1^{\sigma_{ij}}(\vec{c}) = \sum_{i,j=1}^N
\sigma_{ij}c_i^2c_j^2(\sum_{k\neq i,j}c_k^2), 
$$
and $s$ is a stabilization parameter in the nonlinear potential. We
note that such a choice meets the {\bf Assumption 3} when $K = 2$
\cite{boyer2014hierarchy}, namely if only a pair of two fluid phases
is present in the system, the $N$-phase Allen-Cahn and Cahn-Hilliard
equations will fully reduce to those for the corresponding two-phase
system. In Section 3.2.2 of \cite{boyer2014hierarchy}, the authors
successfully constructed the nonlinear potential that meets the {\bf
Assumption 3} when $K \leq 3$. However, the construction of the
consistent $N$-phase nonlinear potential is still an open problem.

\subsection{Phase variables with special choice of $\bA$}

Basically, the choice of $\bA$ in our $N$-phase model does not
affect the dynamics of concentrations. In practice, $\bA$ can be
chosen such that the tangent space $T\Sigma$ can easily be
represented. To this end, a convenient choice is 
\begin{equation} \label{equ:choice-A}
\bA = \begin{pmatrix}
1 & & & \\
& \ddots & & \\
& & 1 & \\
1 & \cdots & 1 & 1
\end{pmatrix}, \quad 
\bA^{-1} = \begin{pmatrix}
1 & & & \\
& \ddots & & \\
& & 1 & \\
-1 & \cdots & -1 & 1
\end{pmatrix}, \quad \vec{b} = \vec{0}.
\end{equation}
In this case, the phase variables are 
$$ 
\vec{\phi} = 
\begin{pmatrix}
1 & & & \\
& \ddots & & \\
& & 1 & \\
1 & \cdots & 1 & 1
\end{pmatrix} 
\begin{pmatrix}
c_1 \\ c_2 \\ \vdots \\ c_N
\end{pmatrix} = 
\begin{pmatrix}
c_1 \\ \vdots \\ c_{N-1} \\ 1
\end{pmatrix}.
$$ 
We also have $\vec{d} = \bA^{-1}\vec{1} = (0, \cdots, 0, 1)^T =
\vec{e}_N$ and 
$$ 
T\Sigma = \{\vec{v}~|~v_N = 0\}, \quad \bP = 
\begin{pmatrix}
1 & & & \\
& \ddots & & \\
& & 1 & \\
& & & 0
\end{pmatrix}.
$$ 
Furthermore, in Theorem \ref{thm:SPD}, we have $\vec{L}_{kN} = \vec{a}_k -
\vec{a}_N = \vec{e}_k$. Thus, 
\begin{equation} \label{equ:choice-Lambda}
\tilde{\bLambda} = \begin{pmatrix}
\frac{9}{2}\tilde{\bsigma}^N & 0 \\
0 & 0
\end{pmatrix}.
\end{equation}
By combining \eqref{equ:choice-A} and \eqref{equ:choice-Lambda}, we
obtain the $N$-phase Allen-Cahn equations \eqref{equ:N-AC-PDE} under
the special choice as  
\begin{equation} \label{equ:special-AC}
\frac{9\gamma}{2}\tilde{\bsigma}^N \pd{}{t} 
\begin{pmatrix}
c_1 \\ c_2 \\ \vdots \\ c_{N-1}
\end{pmatrix} - \nabla \cdot (\frac{9\eta}{2} \tilde{\bsigma}^N 
\begin{pmatrix}
\nabla c_1 \\  \nabla c_2 \\ \vdots \\ \nabla c_{N-1}
\end{pmatrix}) + 
\frac{1}{\eta} \begin{pmatrix}
\pd{F}{c_1} - \pd{F}{c_N} \\
\pd{F}{c_2} - \pd{F}{c_N} \\
\vdots \\
\pd{F}{c_{N-1}} - \pd{F}{c_N}
\end{pmatrix} = 0.
\end{equation}

Similarly, we obtain the $N$-phase Cahn-Hilliard equations \eqref{equ:N-CH-PDE}
under the special choice as 
\begin{equation} \label{equ:special-CH}
\left\{
\begin{aligned}
\tilde{\bA}
\pd{}{t} 
\begin{pmatrix}
c_1 \\ c_2 \\ \vdots \\ c_{N-1}
\end{pmatrix} &= \nabla \cdot ( \frac{2M_0}{9}\tilde{\bA}
(\tilde{\bsigma}^N)^{-1} \tilde{\bA}
\begin{pmatrix}
\nabla w_1 \\ \nabla w_2 \\ \vdots \\ \nabla w_{N-1}
\end{pmatrix}
), \\
\tilde{\bA}
\begin{pmatrix}
w_1 \\ w_2 \\ \vdots \\ w_{N-1}
\end{pmatrix} &= 
- \nabla \cdot( 
\frac{9\eta}{2}\tilde{\bsigma}^N \begin{pmatrix}
\nabla c_1 \\ \nabla c_2 \\ \vdots \\ \nabla c_{N-1}
\end{pmatrix} ) + 
\frac{1}{\eta} \begin{pmatrix}
\pd{F}{c_1} - \pd{F}{c_N} \\
\pd{F}{c_2} - \pd{F}{c_N} \\
\vdots \\
\pd{F}{c_{N-1}} - \pd{F}{c_N}
\end{pmatrix},
\end{aligned}
\right.
\end{equation}
where 
$$
\tilde{\bA} := 
\begin{pmatrix}
2 & 1 & \cdots & 1 \\
1 & 2 & \cdots & 1 \\
\vdots & \ddots & \ddots & \vdots \\
1 & \cdots & 1 & 2
\end{pmatrix} \in \R^{(N-1)\times(N-1)}.
$$

We note that in \cite{dong2014efficient, dong2015physical}, where the
pairwise surface tensions are also considered, the dynamics of
concentrations are dependent on the choice of $\bA$. On the other
hand, our model can be viewed as extending the literature by adding
the effect of the pairwise surface tensions. For the two-phase case,
the $N$-phase Allen-Cahn equations \eqref{equ:special-AC} are shown to
be 
$$ 
\frac{9\gamma}{2}\pd{c_1}{t} - \nabla \cdot (\frac{9\eta}{2} \nabla
c_1) + \frac{8}{\eta} c_1 (1-c_1) (1-2c_1) = 0. 
$$ 
Let $\gamma = \eta$ and take the transformation $c_1 =
\frac{1+\phi}{2}$. Therefore, we have 
$$ 
\pd{\phi}{t} - \Delta \phi +
\frac{8}{9\eta^2} (\phi^3 - \phi) = 0,
$$ 
which yields the standard two-phase Allen-Cahn equation when $\epsilon =
\frac{3}{2\sqrt{2}}\eta$, see \eqref{equ:2-lambda-tension}. 

Similarly, when $N=2$, the $N$-phase Cahn-Hilliard equations
\eqref{equ:special-CH} are shown to be 
$$
\left \{
\begin{aligned}
2\pd{c_1}{t} &= \nabla \cdot (\frac{8M_0}{9\sigma} \nabla w_1), \\
2w_1 &= -\nabla \cdot (\frac{9\eta\sigma}{2}  \nabla c_1) +
\frac{8\sigma}{\eta} c_1(1-c_1)(1-2c_1),
\end{aligned}
\right.
$$ 
or 
$$ 
\phi_t + \nabla \cdot\left(\frac{4M_0}{9}\nabla \left(\frac{9\eta}{4}\Delta
    \phi - \frac{2}{\eta}(\phi^3 - \phi)\right) \right) = 0,
$$ 
which yields the standard two-phase Cahn-Hilliard equation 
$$ 
\pd{\phi}{t} + \Delta \left( \epsilon \Delta \phi -
\frac{1}{\epsilon}(\phi^3 - \phi) \right) = 0,
$$ 
when
$\epsilon = \frac{3}{2\sqrt{2}}\eta$ and $M_0 = \frac{3}{2\sqrt{2}}$.

\section{Discretization for the $N$-phase Models} \label{sec:discretization}

In this section, we present some numerical schemes for both $N$-phase
Allen-Cahn and $N$-phase Cahn-Hilliard equations. Because of the
fundamental role that energy law plays in the phase field model, we
will focus on the energy-stable property of the numerical schemes in
the discrete level. 

The time step size is denoted by $k$. Denote the Hessian matrix of
$F(\vec{c})$ as 
\begin{equation} \label{equ:Hessian}
\bH(\cdot) = \frac{\partial^2 F(\cdot)}{\partial \vec{c}^2}. 
\end{equation}

In the two-phase case, it is well-known that the Allen-Cahn equation
satisfies the maximum principle, which is also satisfied for
the Cahn-Hilliard equations for truncated potentials
\cite{caffarelli1995bound}.
The \emph{admissible states} \eqref{equ:admissible} can be regarded as
the generalization of the maximum principle in the two-phase case:
\begin{equation} \label{equ:admissible}
\cA_c = \{\vec{c}\in \R^N~|~ 0\leq c_i \leq 1, \sum_{i=1}^N c_i =
1\}.
\end{equation}
We note that the SPD property of the coefficient matrix is the
critical point for the maximum principle in the two-phase Allen-Cahn
equation.  For the $N$-phase Allen-Cahn and $N$-phase Cahn-Hilliard
equations, we recall that the physical conditions $\vec{c} \in \cA_c$
cannot be guaranteed in our models. 

From the numerical aspect, let $\cA_{c,h}$ be the numerical
admissible states of the concentrations. That is, the numerical
concentrations are allowed to lie only in $\cA_{c,h}$. Then, we
define two constants: 
\begin{equation}
\label{equ:bound-hessian}
L_1 := \max_{\vec{\xi} \in \cA_{c,h}} \left| \lambda_{\max}\left(
       \bH(\vec{\xi}) \right) \right|, \quad 
L_2 := \max_{\vec{\xi} \in \cA_{c,h}} \left| \lambda_{\min}\left( 
       \bH(\vec{\xi}) \right) \right|. 
\end{equation}
We note that both $L_1$ and $L_2$ depend on the pairwise surface
tensions and the stabilization parameter seeing that
\eqref{equ:nonlinear-potential}. 

\subsection{Numerical schemes for $N$-phase Allen-Cahn equations}

In this subsection, we will extend some existing schemes for two-phase
Allen-Cahn equations to the $N$-phase versions. Let $V_h$ denote the
finite element subspace of $H^1(T\Sigma)$. We note again that $\gamma =
\cO(\eta)$ in the $N$-phase Allen-Cahn model is used to render the
model consistent with mean curvature flow for the two-phase model.
Moreover, by virtue of \eqref{equ:Mob-AC} and
\eqref{equ:AC-invarient2}, we have 
$$ 
\bA^{-T}\bM_c^{\rm AC} \bA^{-1} = \tilde{\bLambda}.
$$ 
Therefore, the strong form of the $N$-phase Allen-Cahn equations
\eqref{equ:N-AC-PDE} turn out to be 
\begin{equation} \label{equ:N-AC-PDE-2}
\left\{
\begin{aligned}
\gamma\bP \tilde{\bLambda} \pd{\vec{\phi}}{t} -\nabla \cdot
(\eta \tilde{\bLambda} \nabla\vec{\phi}) + \frac{1}{\eta} \bP\bA^{-T}
\pd{F}{\vec{c}} &= 0, \quad \text{in}~\Omega\times(0,T], \\
(\nabla \vec{\phi}) \nu & = 0, \quad \text{on}~\partial\Omega
\times(0,T].
\end{aligned}
\right.
\end{equation}

\subsubsection{First-order semi-implicit scheme}
The first-order semi-implicit scheme for $N$-phase Allen-Cahn equations
\eqref{equ:N-AC-PDE-2} can be written as 
\begin{equation} \label{equ:semi-NAC}
\left(\frac{\gamma}{k}\tilde{\bLambda} (\vec{\phi}_h^n - \vec{\phi}_h^{n-1}),
\vec{v}_h\right) + (\eta\tilde{\bLambda} \nabla \vec{\phi}_h^n, \nabla
\vec{v}_h) + \left(\frac{1}{\eta}\bA^{-T}
\pd{F}{\vec{c}}(\vec{\phi}_h^{n-1}), \vec{v}_h \right) = 0,
\quad \forall \vec{v}_h \in V_h.
\end{equation}
Define 
$$ 
G(\vec{\phi}) := \int_{\Omega} \big( \frac{\eta}{2}
\tilde{\bLambda}\nabla \vec{\phi} \big) : \nabla \vec{\phi},
$$ 
which can be easily verified to be convex on $T\Sigma$ thanks to the
SPD property of $\tilde{\bLambda}$.

\begin{theorem} \label{thm:semi-NAC}
For \eqref{equ:semi-NAC}, under the condition that 
$$ 
k \leq \frac{2\lambda_{c,\min}}{L_1}\gamma\eta,
$$ 
the following discrete energy-stability holds: 
$$ 
E(\vec{\phi}_h^{n}) +
(\frac{\gamma\lambda_{c,\min}}{k}-\frac{L_1}{2\eta})\|\vec{c}_h^n -
    \vec{c}_h^{n-1}\|_0^2 \leq E(\vec{\phi}_h^{n-1}), 
$$ 
where $\lambda_{c,\min}$ is the minimal eigenvalue of
$\tilde{\bLambda}_c$ on $T\Sigma_c$. 
\end{theorem}
\begin{proof}
In light of the convexity of $G(\cdot)$, we have 
$$ 
\begin{aligned}
G(\vec{\phi}_h^n) - G(\vec{\phi}_h^{n-1}) & \leq
G'(\vec{\phi}_h^n)(\vec{\phi}_h^n - \vec{\phi}_h^{n-1}) =
\left( \eta 
\tilde{\bLambda}\nabla \vec{\phi}_h^n, \nabla (\vec{\phi}_h^n -
  \vec{\phi}_h^{n-1}) \right) \\
& = - \left( \frac{1}{\eta} \bA^{-T}\pd{F}{\vec{c}}(\vec{\phi}_h^{n-1}), 
    \vec{\phi}_h^n - \vec{\phi}_h^{n-1}\right) -
\left( \frac{\gamma}{k} \tilde{\bLambda}(\vec{\phi}_h^n -
      \vec{\phi}_h^{n-1}), \vec{\phi}_h^n - \vec{\phi}_h^{n-1} \right)
\\
& = - \left( \frac{1}{\eta}\pd{F}{\vec{c}}(\vec{\phi}_h^{n-1}),
    \vec{c}_h^n - \vec{c}_h^{n-1} \right) -
\left(\frac{\gamma}{k} \tilde{\bLambda}_c (\vec{c}_h^n -
      \vec{c}_h^{n-1}), \vec{c}_h^n - \vec{c}_h^{n-1}\right) \\
& = -\left( \frac{1}{\eta}F(\vec{\phi}_h^n), 1 \right) + 
    \left( \frac{1}{\eta}F(\vec{\phi}_h^{n-1}), 1 \right) \\ 
& ~~~ - \left([\frac{\gamma}{k}\tilde{\bLambda}_c -
    \frac{1}{2\eta}\bH(\vec{\xi})] (\vec{c}_h^n - \vec{c}_h^{n-1}),
    \vec{c}_h^n - \vec{c}_h^{n-1} \right).
\end{aligned}
$$ 
The last equality is derived by the Taylor expansion of $F(\cdot)$
around $\vec{\phi}_h^{n-1}$: 
$$ 
F(\vec{\phi}_h^n) - F(\vec{\phi}_h^{n-1}) =
\pd{F(\vec{\phi}_h^{n-1})}{\vec{c}} \cdot (\vec{c}_h^n -
\vec{c}_h^{n-1}) + \frac{1}{2} (\vec{c}_h^n -
  \vec{c}_h^{n-1})^T \bH(\xi)(\vec{c}_h^n - \vec{c}_h^{n-1}).
$$  
Then, we have 
$$ 
E(\vec{\phi}_h^n) - E(\vec{\phi}_h^{n-1}) \leq -(\frac{\gamma
\lambda_{c,\min}}{k} - \frac{L_1}{2\eta})\|\vec{c}_h^n -
\vec{c}_h^{n-1}\|_0^2.
$$ 
This completes the proof.
\end{proof}

\subsubsection{First-order fully-implicit scheme}
The first-order fully-implicit scheme for $N$-phase Allen-Cahn
equations \eqref{equ:N-AC-weak} is:
\begin{equation} \label{equ:fully-NAC}
\left( \frac{\gamma}{k} \tilde{\bLambda}(\vec{\phi}_h^n -
\vec{\phi}_h^{n-1}), \vec{v}_h \right) + 
(\eta\tilde{\bLambda} \nabla \vec{\phi}_h^n, \nabla \vec{v}_h) +
\left(\frac{1}{\eta}\bA^{-T}\pd{F}{\vec{c}}(\vec{\phi}_h^{n}),
\vec{v}_h \right) = 0, \quad \forall \vec{v}_h \in V_h.
\end{equation}

Similar to Theorem \ref{thm:semi-NAC}, we have the following theorem
for the discrete energy-stability and convexity of the fully-implicit
scheme. 

\begin{theorem} \label{thm:fully-NAC}
For \eqref{equ:fully-NAC}, under the condition that 
$$ 
k \leq \frac{2\lambda_{c,\min}}{L_2}\gamma\eta,
$$ 
the following discrete energy-stability holds: 
$$ 
E(\vec{\phi}_h^{n}) + (\frac{\gamma \lambda_{c,\min}}{k} -
    \frac{L_2}{2\eta}) \|\vec{c}_h^n - \vec{c}_h^{n-1}\|_0^2 \leq
E(\vec{\phi}_h^{n-1}). 
$$ 
Furthermore, let 
$$ 
E_1(\vec{\phi}) = \left( \frac{\gamma}{2k} \tilde{\bLambda}(\vec{\phi}
      - \vec{\phi}_h^{n-1}), \vec{\phi} - \vec{\phi}_h^{n-1} \right) +
\int_\Omega \frac{1}{\eta} F(\vec{\phi}) ~\rd x, 
$$
such that \eqref{equ:fully-NAC} can be written as
$E_1'(\vec{\phi}_h^n)(\vec{v}_h) + G'(\vec{\phi}_h^n)(\vec{v}_h) = 0$.
Then, $E_1(\cdot) + G(\cdot)$ is convex when  
$$
k \leq \frac{\lambda_{c,\min}}{L_2}\gamma\eta.
$$ 
\end{theorem}
\begin{proof}
We take the Taylor expansion around $\vec{\phi}_h^{n}$ instead to obtain 
$$ 
F(\vec{\phi}_h^n) - F(\vec{\phi}_h^{n-1}) =
\pd{F(\vec{\phi}_h^{n})}{\vec{c}} \cdot (\vec{c}_h^n -
\vec{c}_h^{n-1}) - \frac{1}{2}  
(\vec{c}_h^n - \vec{c}_h^{n-1})^T \bH(\vec{\xi})(\vec{c}_h^n -
\vec{c}_h^{n-1}).
$$  
In light of the convexity of $G(\cdot)$ again, we have 
$$ 
\begin{aligned}
G(\vec{\phi}_h^n) - G(\vec{\phi}_h^{n-1}) & \leq
G'(\vec{\phi}_h^n)(\vec{\phi}_h^n - \vec{\phi}_h^{n-1}) \\
& = -\left( \frac{1}{\eta}F(\vec{\phi}_h^n), 1\right) + 
\left(\frac{1}{\eta}F(\vec{\phi}_h^{n-1}), 1\right) \\ 
& ~~~ - \left([\frac{\gamma}{k}\tilde{\bLambda}_c + \frac{1}{2\eta}
    \bH(\vec{\xi})] (\vec{c}_h^n - \vec{c}_h^{n-1}), \vec{c}_h^n -
    \vec{c}_h^{n-1}
  \right).
\end{aligned}
$$ 
Thus,
$$ 
E(\vec{\phi}_h^n) - E(\vec{\phi}_h^{n-1}) \leq
-(\frac{\gamma\lambda_{c,\min}}{k}-\frac{L_2}{2\eta}) \|\vec{c}_h^n -
\vec{c}_h^{n-1}\|_0^2.
$$ 

When $k\leq \frac{\lambda_{c,\min}}{L_2}\gamma\eta$ and the Taylor
expansion is applied again, we have 
$$
E_1(\vec{\phi}) - E_1(\vec{\phi}_h^{n-1}) -
E_1'(\vec{\phi})(\vec{\phi} - \vec{\phi}_h^{n-1}) = -
\left([\frac{\gamma}{2k}\tilde{\bLambda}_c + \frac{1}{2\eta}
\bH(\vec{\xi})] \bA^{-1}(\vec{\phi} - \vec{\phi}_h^{n-1}),
\bA^{-1}(\vec{\phi} - \vec{\phi}_h^{n-1}) \right) \leq 0,
$$ 
which means that $E_1(\cdot)$ is convex. Hence, $\vec{\phi}_h^n$ in
\eqref{equ:fully-NAC} is the local minimizer of the convex functional
$E_1(\cdot) + G(\cdot)$ on $V_h$.
\end{proof}


\subsubsection{Modified Crank-Nicolson scheme}
Now we will try to extend the modified Crank-Nicolson scheme
\cite{du1991numerical,condette2011spectral} to the $N$-phase
Allen-Cahn equations. Define the finite difference of $f$ as  
\begin{equation} \label{equ:difference}
f[c, c^*] := 
\begin{cases}
\frac{f(c) - f(c^*)}{c - c^*}, & c \neq c^*, \\
f'(c), & c = c^*.
\end{cases}
\end{equation}
For any set $\bi = \{i_1, i_2, \cdots, i_k\}$ and monomial
$q_{\bi}(\vec{c}) = c_{i_1}c_{i_2} \cdots c_{i_k}$, we define the
finite difference of $q_{\bi}$ as 
\begin{equation} \label{equ:difference-q}
q_{\bi}[\vec{c},\vec{c}^*] = \frac{1}{k!}\sum_{l=1}^k
\left[ \sum_{\bj \subset \bi -\{i_l\}} |\bj|!(k-|\bj|-1)!
q_{\bj}(\vec{c}) q_{\bi-\bj-\{i_l\}}(\vec{c}^*) \right] \vec{e}_{i_l}, 
\end{equation}
where we denote $q_{\varnothing} = 1$.  Then, we have the following
crucial lemma.
\begin{lemma} \label{lem:FD-q} 
It holds that 
\begin{equation}\label{equ:FD-q} 
q_{\bi}(\vec{c}) - q_{\bi}(\vec{c}^*) = q_{\bi}[\vec{c}, \vec{c}^*]
\cdot (\vec{c} - \vec{c}^*).
\end{equation} 
\end{lemma}
\begin{proof}
We will prove it by induction. It is straightforward that
$q_{\bi}[\vec{c},\vec{c}^*] = \vec{e}_{i_1}$ when $k = 1$. Assume
\eqref{equ:FD-q} holds for $|\bi| = k-1$. From the fact that for any
$\bi = \{i_1, i_2, \cdots, i_k\}$,  
$$ 
q_{\bi}(\vec{c}) - q_{\bi}(\vec{c}^*) =
q_{\bi-\{i_s\}}(\vec{c})(c_{i_s} - c_{i_s}^*) + c_{i_s}^*
\left[q_{\bi-\{i_s\}}(\vec{c}) - q_{\bi-\{i_s\}}(\vec{c}^*) \right],
  \quad \forall 1\leq s \leq k,
$$ 
we have 
$$
\begin{aligned}
q_{\bi}(\vec{c}) - q_{\bi}(\vec{c}^*) 
&= \frac{1}{k} \Big\{ \sum_{s=1}^k
q_{\bi-\{i_s\}}(\vec{c})(c_{i_s} - c_{i_s}^*) + c_{i_s}^*
\left[q_{\bi-\{i_s\}}(\vec{c}) - q_{\bi-\{i_s\}}(\vec{c}^*) \right]
\Big\}
 \\
&= \frac{1}{k} \left[ \sum_{s=1}^k
q_{\bi-\{i_s\}}(\vec{c})\vec{e}_{i_s} \right] \cdot (\vec{c} -
    \vec{c}^*)\\ 
& +
\frac{1}{k!} \sum_{s=1}^k \sum_{l=1, l\neq s}^k c_{i_s}^* \left[ 
\sum_{\bj \subset \bi - \{i_l, i_s\}} |\bj|!(k-|\bj|-2)!
q_{\bj}(\vec{c}) q_{\bi-\bj-\{i_l, i_s\}}(\vec{c}^*) \right]
\vec{e}_{i_l} \cdot (\vec{c} - \vec{c}^*)
\end{aligned}
$$ 
Notice that 
$$ 
\begin{aligned}
& \frac{1}{k!} \sum_{s=1}^k \sum_{l=1, l\neq s}^k c_{i_s}^* \left[ 
\sum_{\bj \subset \bi - \{i_l, i_s\}} |\bj|!(k-|\bj|-2)!
q_{\bj}(\vec{c}) q_{\bi-\bj-\{i_l, i_s\}}(\vec{c}^*) \right]
\vec{e}_{i_l} \\
=~& \frac{1}{k!} \sum_{l=1}^k \sum_{s=1, s\neq l}^k \left[ 
\sum_{\bj \subset \bi - \{i_l, i_s\}} |\bj|!(k-|\bj|-2)!
q_{\bj}(\vec{c}) q_{\bi-\bj-\{i_l\}}(\vec{c}^*) \right]
\vec{e}_{i_l} \\
=~& \frac{1}{k!} \sum_{l=1}^k \left[ 
\sum_{\bj \subset \bi - \{i_l\}, \bj \neq \bi-\{i_l\}} |\bj|!(k-|\bj|-1)!
q_{\bj}(\vec{c}) q_{\bi-\bj-\{i_l\}}(\vec{c}^*) \right]
\vec{e}_{i_l}. \\
\end{aligned}
$$ 
Therefore, we have \eqref{equ:FD-q} when $|\bi|=k$. This completes the
proof.
\end{proof}

In light of the above lemma, we define the finite difference of the
nonlinear potential $F$ as follows. For the homogeneous case
\eqref{equ:homo-F}, 
\begin{equation} \label{equ:homo-FD-F}
F^{\sigma}[\vec{c}, \vec{c}^*] := 2\sigma \begin{pmatrix} 
f[c_1, c_1^*] \\
f[c_2, c_2^*] \\
\vdots \\ 
f[c_N, c_N^*]
\end{pmatrix} + 8\sigma \sum_{i_1<i_2<i_3<i_4}
q_{\{i_1,i_2,i_3,i_4\}}[\vec{c}, \vec{c}^*]. 
\end{equation} 
For the inhomogeneous case \eqref{equ:nonlinear-potential},
\begin{equation} \label{equ:inhomo-FD-F} 
F^{\sigma_{ij}}[\vec{c}, \vec{c}^*] := 
\begin{pmatrix}
\sum_{j=1}^N \sigma_{1j} (f[c_1,c_1^*] - f[c_1+c_j, c_1^*+c_j^*]) \\
\sum_{j=1}^N \sigma_{2j} (f[c_2,c_2^*] - f[c_2+c_j, c_2^*+c_j^*]) \\
\vdots \\
\sum_{j=1}^N \sigma_{Nj} (f[c_N,c_N^*] - f[c_N+c_j, c_N^*+c_j^*]) \\
\end{pmatrix} + s \sum_{i,j=1}^N \sum_{k \neq i, j}
\sigma_{ij}(c_ic_jc_k + c_i^*c_j^*c_k^*)
  q_{\{i,j,k\}}[\vec{c},\vec{c}^*].
\end{equation} 
Let $F = F^{\sigma}$ or $F^{\sigma{ij}}$.  Then, a routine calculation
shows that 
$$ 
F(c) - F(\vec{c}^*) = F[\vec{c}, \vec{c}^*] \cdot (\vec{c} -
    \vec{c}^*).
$$ 
We, therefore, obtain the following modified Crank-Nicolson scheme:
\begin{equation} \label{equ:Crank-Nicolson-NAC}
(\frac{\gamma}{k}\tilde{\bLambda}(\vec{\phi}_h^n -
\vec{\phi}_h^{n-1}), \vec{v}_h) + 
(\eta \tilde{\bLambda} \nabla \frac{\vec{\phi}_h^n +
 \vec{\phi}_h^{n-1}}{2}, \nabla \vec{v}_h) +
 \left(\frac{1}{\eta}\bA^{-T}F[\vec{c}_h^n, \vec{c}_h^{n-1}],
\vec{v}_h \right) = 0, \quad \forall \vec{v}_h \in V_h.
\end{equation}
Taking $\vec{v}_h = \vec{\phi}_h^{n} - \vec{\phi}_h^{n-1}$, we
immediately obtain the following result:

\begin{theorem} \label{thm:Crank-Nicolson-NAC}
Scheme \eqref{equ:Crank-Nicolson-NAC} is unconditionally
energy-stable, and 
$$ 
E(\vec{\phi}_h^{n}) + \frac{\gamma}{k}\|\vec{c}_h^n -
\vec{c}_h^{n-1}\|_0^2 = E(\vec{\phi}_h^{n-1}).
$$ 
\end{theorem}

This theorem satisfies the unconditionally energy-stability of the
modified Crank-Nicolson scheme. However, it is necessary to solve a
nonlinear system, the existence and uniqueness for which can only be
numerically validated under a condition $k \leq C\eta^2$ for a certain
constant $C>0$. We refer to \cite{condette2011spectral} for proof of
the two-phase case and the numerical tests for the $N$-phase case in
Section \ref{sec:tests}.

\subsection{Numerical schemes for $N$-phase Cahn-Hilliard equations}

We will discuss the numerical schemes for $N$-phase Cahn-Hilliard
equations. In this subsection, we denote $V_h, Q_h$ as the finite
element subspace of $H_1(T\Sigma)$. By virtue of \eqref{equ:Mob-CH},
the strong form of the $N$-phase Cahn-Hilliard equations
\eqref{equ:N-CH-PDE} turn out to be  
\begin{equation} \label{equ:N-CH-PDE-2}
\left\{
\begin{aligned}
\bP\bA^{-T}\bA^{-1} \pd{\vec{\phi}}{t} &= \nabla \cdot [
M_0(\bA^{-1}\bP)^T\tilde{\bLambda}_c^{\dagger}
    \bA^{-1}\bP \nabla \vec{w}], \quad
\text{in}~\Omega\times(0,T], \\
\bP\bA^{-T}\bA^{-1}\vec{w} &= - \nabla \cdot
(\eta \tilde{\bLambda}\nabla \vec{\phi}) +
\frac{1}{\eta}\bP\bA^{-1}\pd{F}{\vec{c}}, \quad
\text{in}~\Omega \times(0,T], \\
(\nabla\vec{\phi})\nu = (\nabla\vec{w})\nu &= 0, \quad
\text{on}~\partial\Omega\times(0,T].
\end{aligned}
\right.
\end{equation}

\subsubsection{First-order semi-implicit scheme}

To make the scheme energy stable, we give the following first-order
semi-implicit scheme for $N$-phase Cahn-Hilliard equations:
\begin{equation} \label{equ:semi-NCH}
\left\{
\begin{aligned}
(\bA^{-1}(\vec{\phi}_h^n - \vec{\phi}_h^{n-1}), \bA^{-1}\vec{q}_h) +
k(M_0 \tilde{\bLambda}_c^{\dagger} \bA^{-1}\nabla \vec{w}_h^n,
  \bA^{-1}\nabla \vec{q}_h) &= 0, \quad \forall \vec{q}_h \in Q_h, \\
 -(\bA^{-1}\vec{w}_h^n,\bA^{-1}\vec{v}_h) +
(\eta \tilde{\bLambda} \nabla\vec{\phi}_h^n, \nabla \vec{v}_h) +
\left(\frac{1}{\eta}\bA^{-T}
\pd{F(\vec{\phi}_h^{n-1})}{\vec{c}}, \vec{v}_h\right) &= 0, \quad
\forall \vec{v}_h \in V_h.
\end{aligned}
\right.
\end{equation}

\begin{theorem} \label{thm:semi-NCH}
For \eqref{equ:semi-NCH}, if $Q_h \subset V_h$, then energy-stability
holds when 
\begin{equation} \label{equ:semi-NCH-stab2}
k \leq \frac{8\lambda_{c,\min}^2}{M_0 L_1^2}\eta^3,
\end{equation}
where $\lambda_{c,\min}$ is the minimal eigenvalue of
$\tilde{\bLambda}_c$ on $T\Sigma_c$. 
\end{theorem}

\begin{proof}
Taking $\vec{q}_h = \vec{w}_h$ and $\vec{v}_h = \vec{\phi}_h^n -
\vec{\phi}_h^{n-1}$ in \eqref{equ:semi-NCH}, we obtain 
$$ 
k(M_0 \tilde{\bLambda}_c^{\dagger} \bA^{-1}\nabla \vec{w}_h^n,
\bA^{-1} \nabla \vec{w}_h^n)  + (\eta \tilde{\bLambda}\nabla
\vec{\phi}_h^n, \nabla(\vec{\phi}_h^n - \vec{\phi}_h^{n-1})) +
  \left(\frac{1}{\eta} \pd{F(\vec{\phi}_h^{n-1})}{\vec{c}},
      \vec{c}_h^n - \vec{c}_h^{n-1} \right) = 0.
$$ 
With the help of the Taylor expansion around $\vec{\phi}_h^{n-1}$ and
the fact that 
$$ 
(\tilde{\bLambda} \nabla \vec{\phi}_h^n, \nabla (\vec{\phi}_h^n -
\vec{\phi}_h^{n-1})) = \frac{1}{2}(\tilde{\bLambda} \nabla \vec{\phi}_h^n,
\nabla \vec{\phi}_h^n) - \frac{1}{2}(\tilde{\bLambda} \nabla
\vec{\phi}_h^{n-1}, \nabla \vec{\phi}_h^{n-1}) +
\frac{1}{2}(\tilde{\bLambda} \nabla (\vec{\phi}_h^n -
  \vec{\phi}_h^{n-1}), \nabla (\vec{\phi}_h^n - \vec{\phi}_h^{n-1})), 
$$ 
we have 
\begin{equation} \label{equ:semi-NCH-bound} 
\begin{aligned}
E(\vec{\phi}_h^n) - E(\vec{\phi}_h^{n-1}) =& \left( 
\frac{1}{2\eta}\bH(\vec{\xi})(\vec{c}_h^n -
\vec{c}_h^{n-1}), \vec{c}_h^n - \vec{c}_h^{n-1}
\right) \\
& -k(M_0 \tilde{\bLambda}_c^{\dagger} \bA^{-1} \nabla\vec{w}_h^n,
    \bA^{-1}\nabla \vec{w}_h^n) 
-\left(\frac{\eta}{2}\tilde{\bLambda} \nabla (\vec{\phi}_h^n - \vec{\phi}_h^{n-1}),
      \nabla (\vec{\phi}_h^n - \vec{\phi}_h^{n-1}) \right) \\
\leq & ~\frac{L_1}{2\eta}\|\vec{c}_h^n -
  \vec{c}_h^{n-1}\|_0^2 \\ 
& -k\left(M_0 \tilde{\bLambda}_c^{\dagger}  \nabla \vec{w}_{c,h}^n, \nabla
    \vec{w}_{c,h}^n \right)
  -\left( \frac{\eta}{2}\tilde{\bLambda}_c \nabla (\vec{c}_h^n - \vec{c}_h^{n-1}),
      \nabla (\vec{c}_h^n - \vec{c}_h^{n-1}) \right),
\end{aligned}
\end{equation}
where $\vec{w}_{c,h} = \bA^{-1} \vec{w}_h \in H^1(T\Sigma_c)$.  If
$Q_h \subset V_h$, then $\vec{q}_h$ can be taken as $\vec{\phi}_h^n -
\vec{\phi}_h^{n-1}$ in order to obtain 
$$ 
(M_0 \tilde{\bLambda}_c^{\dagger} \nabla \vec{w}_{c,h}^n, \nabla (\vec{c}_h^n -
\vec{c}_h^{n-1})) = -\frac{1}{k} \|\vec{c}_h^n -
\vec{c}_h^{n-1}\|_0^2.
$$ 
Then, 
$$ 
\begin{aligned}
& k \left( M_0 \tilde{\bLambda}_c^{\dagger} \nabla \vec{w}_{c,h}^n, \nabla
    \vec{w}_{c,h}^n \right) +
\left( \frac{\eta}{2}\tilde{\bLambda}_c \nabla (\vec{c}_h^n - \vec{c}_h^{n-1}),
    \nabla (\vec{c}_h^n - \vec{c}_h^{n-1}) \right) \\ 
= & ~ \frac{k}{M_0} \left( \tilde{\bLambda}_c (M_0
      \tilde{\bLambda}_c^{\dagger} \nabla \vec{w}_{c,h}^n), (M_0
     \tilde{\bLambda}_c^{\dagger} \nabla \vec{w}_{c,h}^n) \right) +
\left( \frac{\eta}{2} \tilde{\bLambda}_c \nabla (\vec{c}_h^n - \vec{c}_h^{n-1}),
    \nabla (\vec{c}_h^n - \vec{c}_h^{n-1}) \right) \\ 
\geq & 
-2\lambda_{c,\min}\sqrt{\frac{\eta k}{2M_0}} \left( M_0
      \tilde{\bLambda}_c^{\dagger} \nabla
\vec{w}_{c,h}^n, \nabla (\vec{c}_h^n - \vec{c}_h^{n-1}) \right) 
= \lambda_{c,\min}\sqrt{\frac{2\eta}{M_0k}}  \|\vec{c}_h^n -
\vec{c}_h^{n-1}\|_0^2.
\end{aligned}
$$ 
Then, from \eqref{equ:semi-NCH-bound}, 
$$ 
E(\vec{\phi}_h^n) - E(\vec{\phi}_h^{n-1}) \leq - (
  \lambda_{c,\min}\sqrt{\frac{2\eta}{M_0k}} - \frac{L_1}{2\eta}) 
\|\vec{c}_h^n - \vec{c}_h^{n-1}\|_0^2,
$$ 
which gives rise to the energy-stability when
\eqref{equ:semi-NCH-stab2} holds.
\end{proof}

\subsubsection{Some nonlinear schemes}
By applying the similar idea of the fully-implicit scheme for $N$-phase
Allen-Cahn equations, we have the following first-order fully-implicit
scheme for $N$-phase Cahn-Hilliard equations:
\begin{equation} \label{equ:fully-NCH}
\left\{
\begin{aligned}
(\bA^{-1}(\vec{\phi}_h^n - \vec{\phi}_h^{n-1}), \bA^{-1}\vec{q}_h) +
k(M_0 \tilde{\bLambda}_c^{\dagger} \bA^{-1}\nabla \vec{w}_h^n,
  \bA^{-1}\nabla \vec{q}_h) &= 0, \quad \forall \vec{q}_h \in Q_h, \\
 -(\bA^{-1}\vec{w}_h^n,\bA^{-1}\vec{v}_h) +
(\eta \tilde{\bLambda} \nabla\vec{\phi}_h^n, \nabla \vec{v}_h) +
\left(\frac{1}{\eta}\bA^{-T}
\pd{F(\vec{\phi}_h^{n})}{\vec{c}}, \vec{v}_h\right) &= 0, \quad
\forall \vec{v}_h \in V_h.
\end{aligned}
\right.
\end{equation}
And, the following theorem can be proved by slightly modifying the
proof of Theorem \ref{thm:semi-NCH}.

\begin{theorem} \label{thm:fully-NCH}
For \eqref{equ:fully-NCH}, if $Q_h \subset V_h$, then energy-stability
holds when 
\begin{equation} \label{equ:fully-NCH-stab2}
k \leq \frac{8\lambda_{c,\min}^2}{M_0 L_2^2}\eta^3.
\end{equation}
\end{theorem}

Another naturally extended scheme for $N$-phase Cahn-Hilliard
equations is the modified Crank-Nicolson scheme:
\begin{equation} \label{equ:Crank-Nicolson-NCH}
\left\{
\begin{aligned}
(\bA^{-1}(\vec{\phi}_h^n - \vec{\phi}_h^{n-1}), \bA^{-1}\vec{q}_h) +
k(M_0 \tilde{\bLambda}_c^{\dagger} \bA^{-1}\nabla \vec{w}_h^n,
  \bA^{-1}\nabla \vec{q}_h) &= 0, \quad \forall \vec{q}_h \in Q_h, \\
 -(\bA^{-1}\vec{w}_h^n,\bA^{-1}\vec{v}_h) +
(\eta \tilde{\bLambda}
 \nabla\frac{\vec{\phi}_h^n+\vec{\phi}_h^{n-1}}{2}, \nabla \vec{v}_h) +
\left(\frac{1}{\eta}\bA^{-T}
F[\vec{c}_h^n,\vec{c}_h^{n-1}], \vec{v}_h\right) &= 0, \quad
\forall \vec{v}_h \in V_h,
\end{aligned}
\right.
\end{equation}
where $F[\cdot,\cdot]$ is the finite difference of the nonlinear
potential defined in \eqref{equ:homo-FD-F} and
\eqref{equ:inhomo-FD-F}, regarding to the homogeneous and
inhomogeneous case, respectively.  The following energy-stability can
be proved, as expected.
\begin{theorem} \label{thm:Crank-Nicolson-NCH}
Scheme \eqref{equ:Crank-Nicolson-NCH} is unconditionally
energy-stable.
\end{theorem}
\begin{proof}
Taking $\vec{v}_h = \vec{\phi}_h^n - \vec{\phi}_h^{n-1}$ and
$\vec{q}_h = \vec{w}_h^n$ in \eqref{equ:Crank-Nicolson-NCH}, we have 
$$
E(\vec{\phi}_h^n) - E(\vec{\phi}_h^{n-1}) = (\bA^{-1}\vec{w}_h^n,
\bA^{-1}(\vec{\phi}_h^n - \vec{\phi}_h^{n-1})) = -k(M_0
\tilde{\bLambda}_c^{\dagger} \bA^{-1}\nabla\vec{w}_h^n,
\bA^{-1}\nabla\vec{w}_h^n) \leq 0.
$$ 
This completes the proof.
\end{proof}

We note that nonlinear schemes \eqref{equ:fully-NCH} and
\eqref{equ:Crank-Nicolson-NCH} for $N$-phase Cahn-Hilliard equations
require the nonlinear solver at each time step, the convergence of which
is difficult to verify. Intuitively, one needs to balance the energy
stability of the numerical scheme and the convergence of the solver at
each time step. In the numerical tests, we will focus on the
semi-implicit scheme for $N$-phase Cahn-Hilliard equations.


\section{Numerical Results} \label{sec:tests}

In this section, we introduce a series of numerical experiments to
illustrate the characteristics of the schemes for our $N$-phase model.
With the special choice of $\bA$ in \eqref{equ:choice-A}, we know that 
$$ 
H^1(T\Sigma) = \{\vec{v} \in H^1(\R^N)~|~ v_N = 0\},
$$ 
which can be discretized by the piecewise linear Lagrangian element
for the first $N-1$ components.  Suppose that the domain is subdivided
by a shape-regular simplicial grid $\mathcal{T}_h = \{K\}$. Then, in the
numerical experiments, we apply 
$$ 
V_h = Q_h = \{\vec{v}_h \in H^1(\R^N)~|~ v_i|_K \in \mathcal{P}_1(K)
, 1\leq i \leq N-1, v_N = 0\}.
$$


\subsection{$N$-phase Allen-Cahn: Grain growth on the unit square domain}
In order to validate the numerical algorithm for $N$-phase Allen-Cahn
equations, we consider the grain growth on the unit square domain
$\Omega = [0,1]\times [0,1]$ with $N = 5$. The uniform mesh with $h =
1/256$ is used for computation. The initial condition here is a
randomly chosen superposition of $1,000$ circular grains, whose radii
range from $0.01$ to $0.04$. We set the characteristic scale of the
interfacial thickness $\eta = 0.005$ and the interfacial surface
tension $\sigma_{ij} = 1$ so that the coefficient matrix
$\tilde{\bLambda} = \cO(1)$ is computed by
\eqref{equ:coef2-equations-matrix}. The nonlinear potential is chosen
as \eqref{equ:homo-F}. The parameter is $\gamma = \eta$. 

First, we illustrate the energy-stability of different
schemes.  The time step size is chosen as $k = 2\times 10^{-5}$ for
the semi-implicit scheme \eqref{equ:semi-NAC}, the fully-implicit
scheme \eqref{equ:fully-NAC}, and the modified Crank-Nicolson scheme
\eqref{equ:Crank-Nicolson-NAC}. The initial conditions for these 
schemes are the same. For the nonlinear scheme in each time step, the
numerical solution on previous step $\vec{\phi}_h^{n-1}$ is used as
the initial guess, and the standard Newton solver is applied with the
stopping criteria that the residual is less than $10^{-5}$ times the
initial residual. Figure \ref{fig:NAC-grain-phases} shows the initial
condition and evolution of the phases computed by the modified
Crank-Nicolson scheme. Similar to the results in
\cite{vanherpe2010multigrid, lee2012practically}, we observe fast
separation in the beginning and slower dynamics in the course of the
evolution. The evolution of the Liapunov free-energy for each scheme is 
depicted in Figure \ref{fig:NAC-grain-energy}. All the
schemes can be observed to be energy-stable, and the respective
dissipation rates of the fully-implicit and modified Crank-Nicolson
schemes are very similar. 

\begin{figure}[!htbp]
\centering 
\captionsetup{justification=centering}
\subfloat[$t=0$]{
  \includegraphics[width=0.3\textwidth]{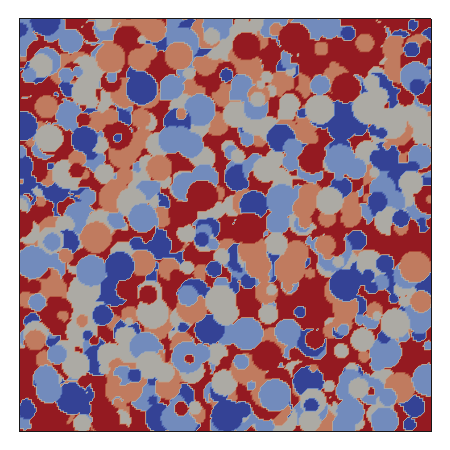} 
}%
\subfloat[$t=24k=4.8\times 10^{-4}$]{
  \includegraphics[width=0.3\textwidth]{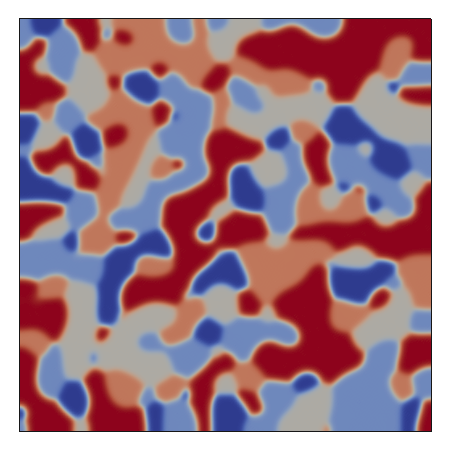} 
}%
\subfloat[$t=45k=9\times 10^{-4}$]{
  \includegraphics[width=0.3\textwidth]{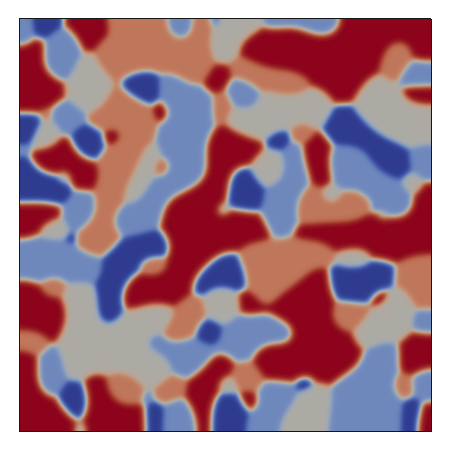}
}\\
\subfloat[$t=105k=2.1\times 10^{-3}$]{
  \includegraphics[width=0.3\textwidth]{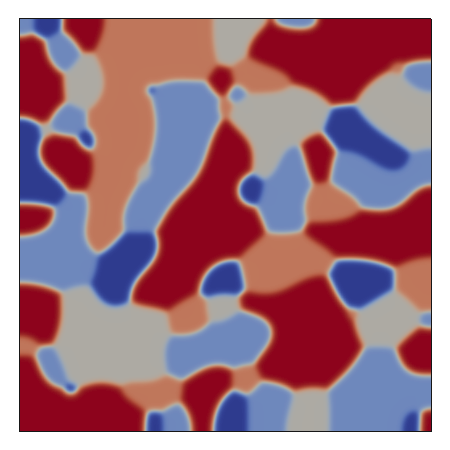} 
}%
\subfloat[$t=225k=4.5\times 10^{-3}$]{
  \includegraphics[width=0.3\textwidth]{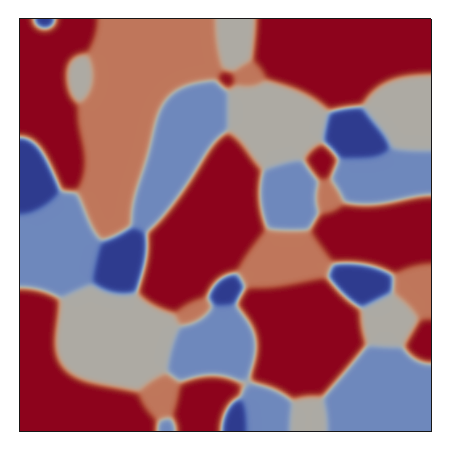} 
}%
\subfloat[$t=450k=9\times 10^{-3}$]{
  \includegraphics[width=0.3\textwidth]{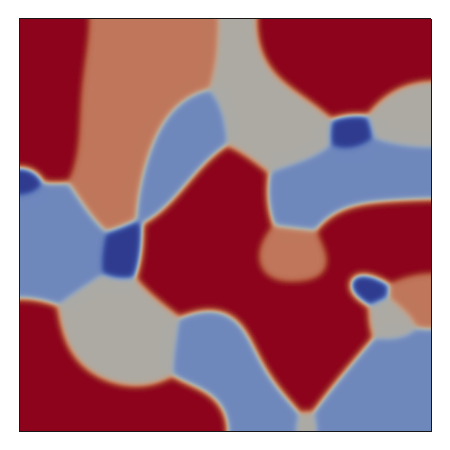}
}
\caption{$N$-phase Allen-Cahn equations: Evolution of the phases for
$N = 5$ by modified Crank-Nicolson scheme}
\label{fig:NAC-grain-phases}
\end{figure}

\begin{figure}[!htbp]
\centering 
\captionsetup{justification=centering}
\includegraphics[width=0.6\textwidth]{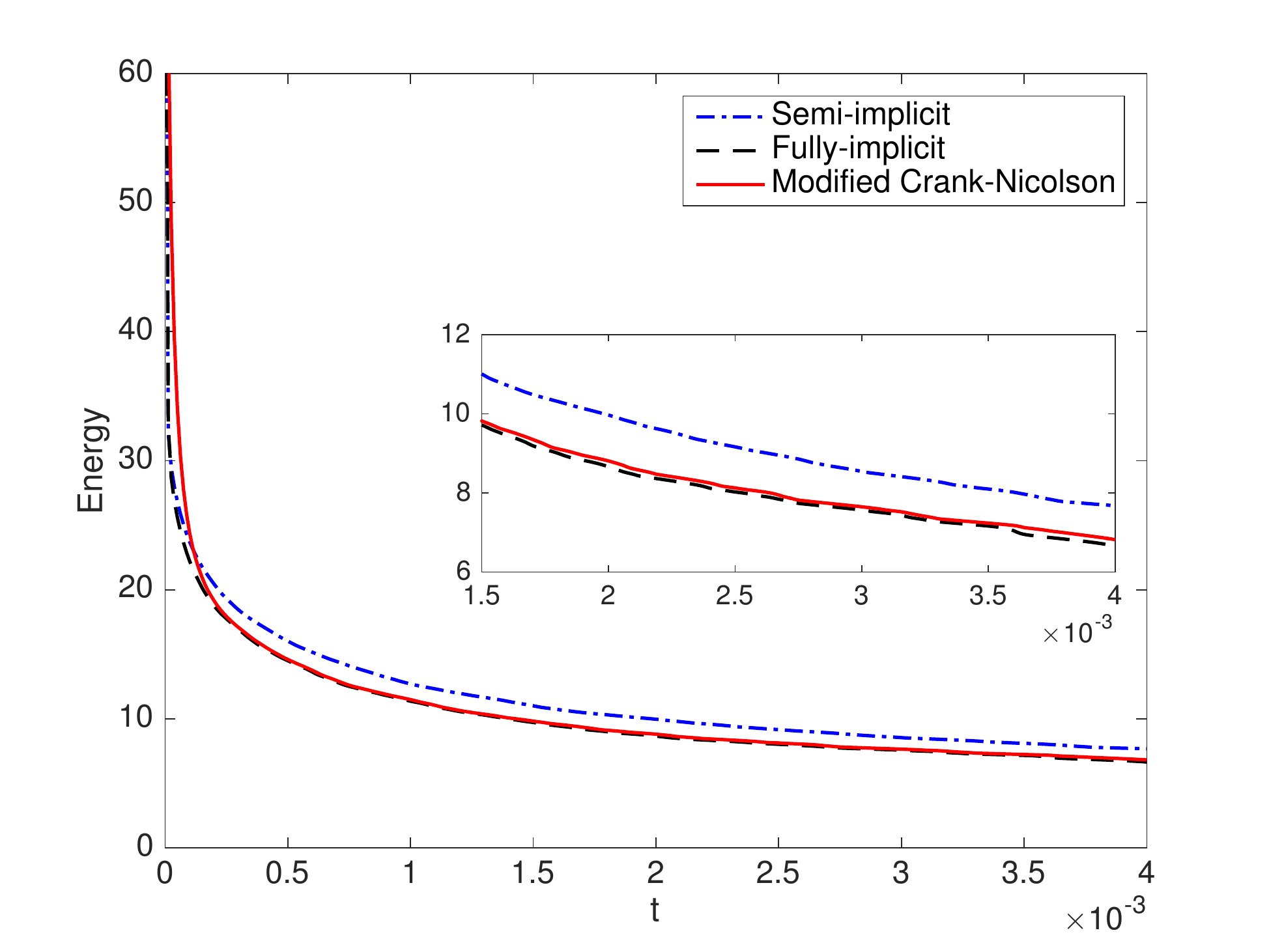}
\caption{$N$-phase Allen-Cahn equations: Evolution of energy for
$N = 5$}
\label{fig:NAC-grain-energy}
\end{figure}

\subsection{$N$-phase Cahn-Hilliard: Spinodal decomposition -- the phase
separation of a three-component mixture}

The second numerical experiment is the phase separation of a three-phase
mixture by spinodal decomposition. Similar tests are also studied in
\cite{lee2008second, lee2012practically, graser2014nonsmooth}. The
initial conditions are random perturbations of state $\vec{c} =
\vec{\rho}$ with the maximum amplitude of $0.04$, that is, 
$$ 
\vec{c} = 
\begin{pmatrix}
\rho_1 + 0.06(2\xi_1 - \xi_2 - \xi_3)/3 \\
\rho_2 + 0.06(-\xi_1 + 2\xi_2 - \xi_3)/3 \\
\rho_3 + 0.06(-\xi_1 - \xi_2 + 2\xi_3)/3 
\end{pmatrix},
$$ 
where $\xi_i \sim \mathcal{U}[0,1]$ are the random variables that obey
the uniform distribution. A $160\times 160\times 2$ uniform triangular
grid is used on the computational domain $\Omega = [0,1] \times
[0,1]$. We take $\eta = 0.01$, $M_0 = \frac{3}{2\sqrt{2}}$, and the
nonlinear potential as \eqref{equ:nonlinear-potential} with $s=0$. The
time step size is set to be $k = 1\times 10^{-6}$. 

In the first three tests, the homogeneous surface tension
$\sigma_{ij} = 1$ is applied with different states $\vec{\rho}$. For
the uniform state $\vec{\rho} = (\frac{1}{3}, \frac{1}{3},
\frac{1}{3})^T$, the result is presented in Figure
\ref{fig:NCH-spinodal-homo-1}. As expected, the three phases have
similar dynamics evolution, as the pairwise surface tensions and
composition are completely symmetric with respect to the four phases.
When the initial state is non-uniform, spinodal decomposition takes
place and the system separates into spatial regions rich in some
phases and poor in others.  For $\vec{\rho} = (\frac{1}{4},
\frac{1}{4}, \frac{1}{2})^T$, the early states of spinodal
decomposition are observed in Figure \ref{fig:NCH-spinodal-homo-2}.
When $\vec{\rho} = (\frac{1}{5}, \frac{1}{5}, \frac{3}{5})^T$, the
phase 3 (blue) in Figure \ref{fig:NCH-spinodal-homo-3}, dominates the
evolution, which leads to spinodal decomposition.   

Thanks to our generalized multiphase models, we are able to
simulate the spinodal decomposition for the inhomogeneous surface
tension case. Here, we set $\sigma_{13} = 1.69$ and the others are
$\sigma_{ij} = 1$. As shown in Figure
\ref{fig:NCH-spinodal-inhomo-1}--\ref{fig:NCH-spinodal-inhomo-3},
the phase 1 (red) and phase 3 (blue) tend to repel each other due to
the relatively large surface tension of each.  

\begin{figure}[!htbp]
\centering 
\captionsetup{justification=centering}
\subfloat[$\sigma_{ij}=1, \vec{\rho} =
(\frac{1}{3},\frac{1}{3},\frac{1}{3})^T$]{
  \includegraphics[width=0.3\textwidth]{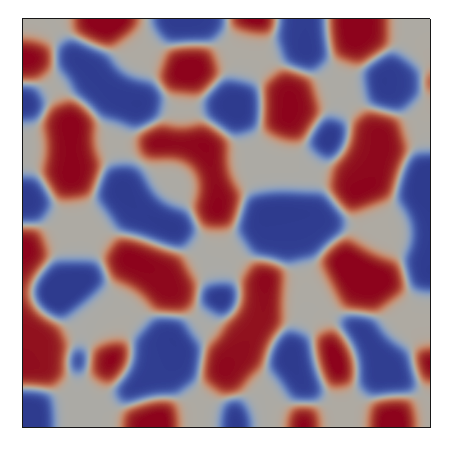} 
  \label{fig:NCH-spinodal-homo-1}
}%
\subfloat[$\sigma_{ij}=1, \vec{\rho} =
(\frac{1}{4},\frac{1}{4},\frac{1}{2})^T$]{
  \includegraphics[width=0.3\textwidth]{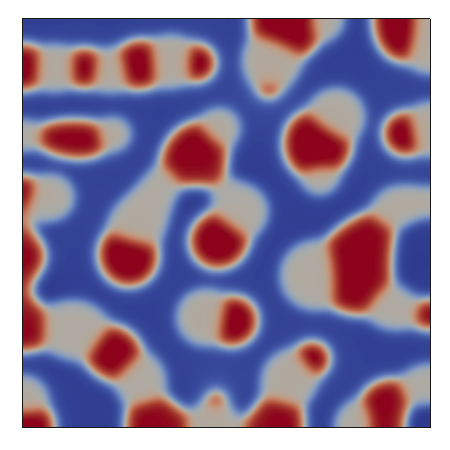} 
  \label{fig:NCH-spinodal-homo-2}
}%
\subfloat[$\sigma_{ij}=1, \vec{\rho} =
(\frac{1}{5},\frac{1}{5},\frac{3}{5})^T$]{
  \includegraphics[width=0.3\textwidth]{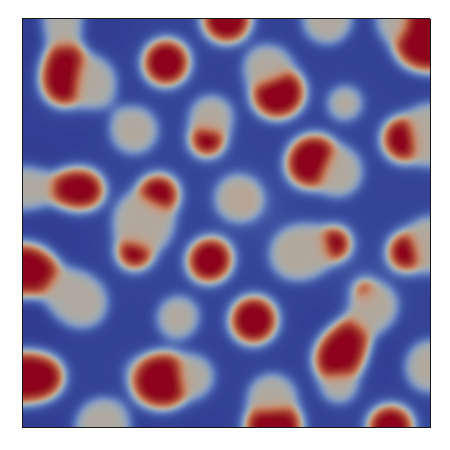} 
  \label{fig:NCH-spinodal-homo-3}
}\\
\subfloat[$\sigma_{13}=1.69, \vec{\rho} =
(\frac{1}{3},\frac{1}{3},\frac{1}{3})^T$]{
  \includegraphics[width=0.3\textwidth]{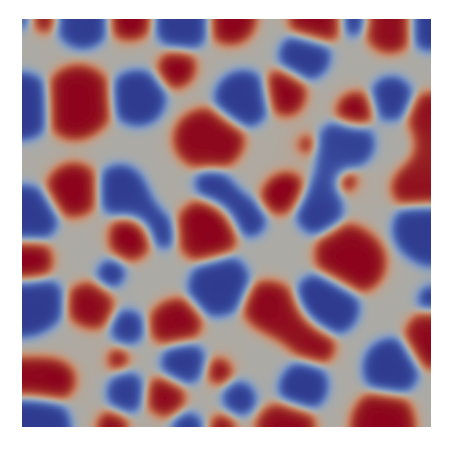} 
  \label{fig:NCH-spinodal-inhomo-1}
}%
\subfloat[$\sigma_{13}=1.69, \vec{\rho} =
(\frac{1}{4},\frac{1}{4},\frac{1}{2})^T$]{
  \includegraphics[width=0.3\textwidth]{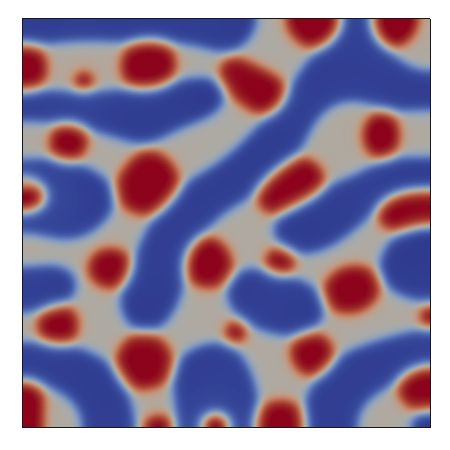} 
  \label{fig:NCH-spinodal-inhomo-2}
}%
\subfloat[$\sigma_{13}=1.69, \vec{\rho} =
(\frac{1}{5},\frac{1}{5},\frac{3}{5})^T$]{
  \includegraphics[width=0.3\textwidth]{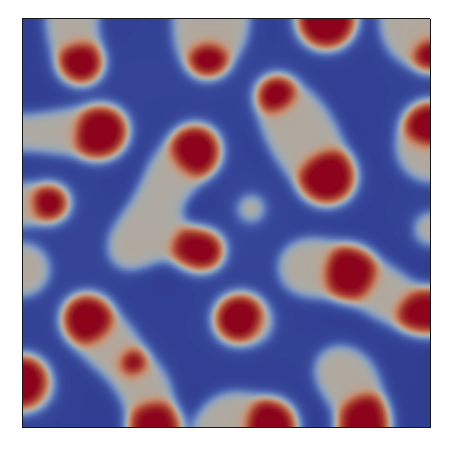} 
  \label{fig:NCH-spinodal-inhomo-3}
}
\caption{$N$-phase Cahn-Hilliard equations: Evolution of spinodal
decomposition with different surface tensions, $t=500k = 5\times
10^{-4}$}
\label{fig:NCH-spinodal}
\end{figure}

\subsection{$N$-phase Cahn-Hilliard: Triple junctions in a quaternary
system}

The last experiment numerically simulates the evolution of the triple
junctions in a quaternary system. In \cite{lee2008second}, the authors
proposed a test for the homogeneous surface tension case.  Here, we
intend to demonstrate the effect of pairwise surface tensions,
especially for the inhomogeneous case.
 
For all the experiments for triple junctions, we simulate how a
T-shaped triple junction approaches a local equilibrium state under
the effect of pairwise surface tensions. A $100 \times 100 \times
2$ uniform triangular grid is used on the computational domain $\Omega
= [0,1]\times[0,1]$. In the semi-implicit scheme
\eqref{equ:semi-NCH}, the parameters are chosen as $\eta = 0.02$.  The
mobility are set as $M_0 = \frac{3}{2\sqrt{2}}$.  The
initial profile and corresponding coloring are depicted in Figure
\ref{fig:NCH-triple-init-profile} and \ref{fig:NCH-triple-init-color},
respectively.  The solutions are computed until numerically
stationary. Even though the time step size can be set small enough
to guarantee the energy stability, we observe in our experiments that
it may vary according to the current state. In general, when the
phases evolving fast or approaching to the topological change, the
time step size should be set small. Otherwise, it can be set larger
than the theoretical constraint \eqref{equ:semi-NCH-stab2} to speed up the
simulation.  

In Figure \ref{fig:NCH-triple-homo-1}--\ref{fig:NCH-triple-homo-4}, we
display the evolution of the interface for the case in which
$\sigma_{ij} = 1$. The stabilization parameter in the nonlinear
potential \eqref{equ:nonlinear-potential} is set as $s=30$, and the
minimal time step size is set as $k=5\times 10^{-8}$. For this case
with homogeneous surface tension, we observe that the triple junction
angles approach the true value $120^{\circ}$ as they approach a local
equilibrium state. We then compute two inhomogeneous surface tension
cases as follows: 
\begin{itemize}
\item Inhomogeneous case 1:
$$ 
\sigma_{ij} = \begin{cases}
1.69, & (i,j) = (1,2), \\ 
1, & \text{else}. 
\end{cases}
$$ 
\item Inhomogeneous case 2:
$$ 
\sigma_{ij} = \begin{cases}
2.56, & (i,j) = (1,2), \\ 
1, & \text{else}. 
\end{cases}
$$ 
\end{itemize}
It is easy to check that these two sets of surface tensions satisfy
the condition in Theorem \ref{thm:SPD}. Thus, $\tilde{\bLambda}$ is
SPD on the tangent space $T\Sigma$. As can be seen from Figure
\ref{fig:NCH-triple-inhomo-1}--\ref{fig:NCH-triple-inhomo-4}, for the
inhomogeneous case 1, the interface between phases 1 and 2 becomes
smaller and smaller due to the relatively large surface tension.
Moreover, the inhomogeneous case 2 encounters the situation with
$\sigma_{12} > \sigma_{13} + \sigma_{23}$ and $\sigma_{12} >
\sigma_{14} + \sigma_{24}$, which corresponds to the total wetting
\cite{de2013capillarity} that the phase 1 and 2 will be penetrated by
phase 3 and 4, as shown in Figure
\ref{fig:NCH-triple-inhomo2-1}--\ref{fig:NCH-triple-inhomo2-4}. 

\begin{figure}[!htbp]
\centering 
\captionsetup{justification=centering}
\subfloat[initial profile]{
  \includegraphics[width=0.25\textwidth]{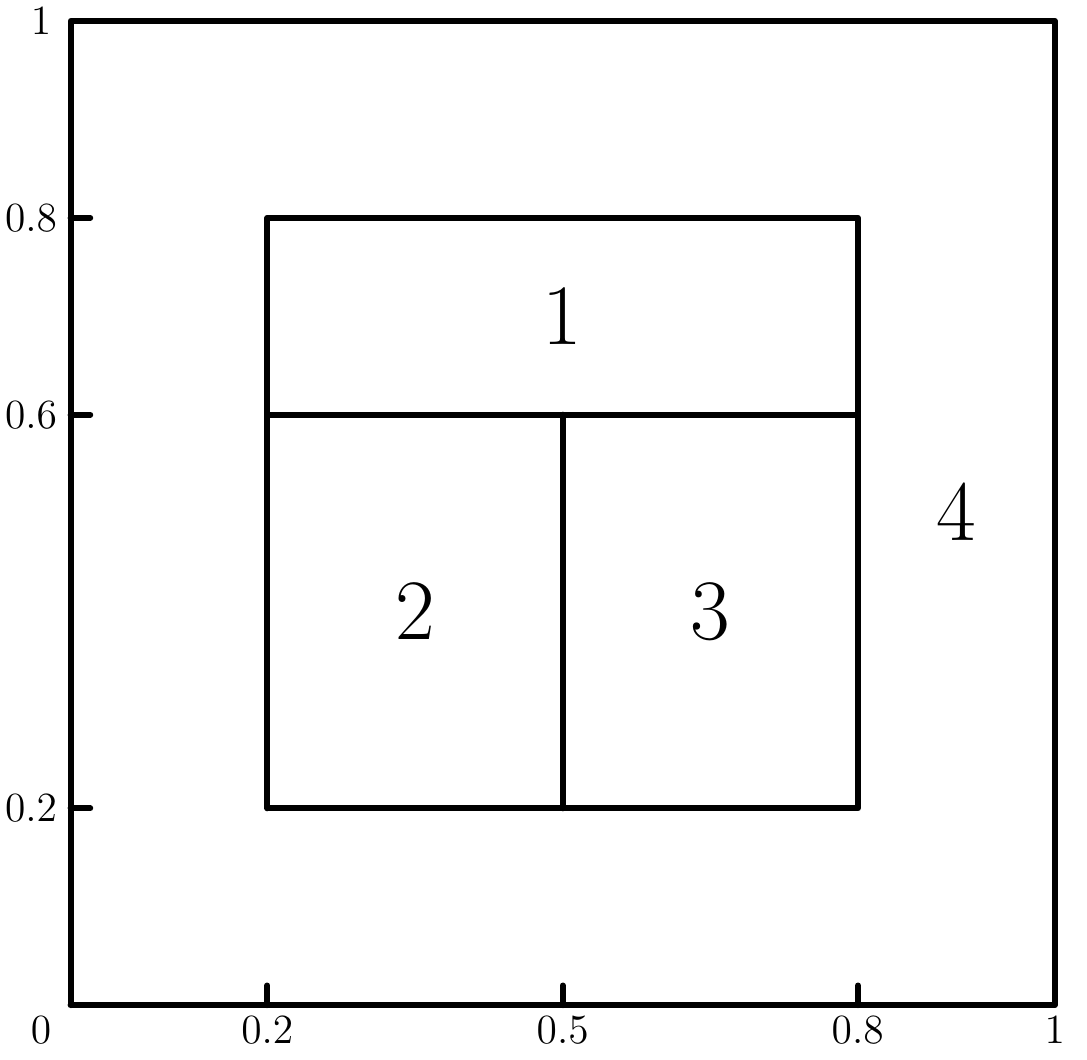} 
  \label{fig:NCH-triple-init-profile}
}%
\subfloat[colors for different phases]{
  \includegraphics[width=0.25\textwidth]{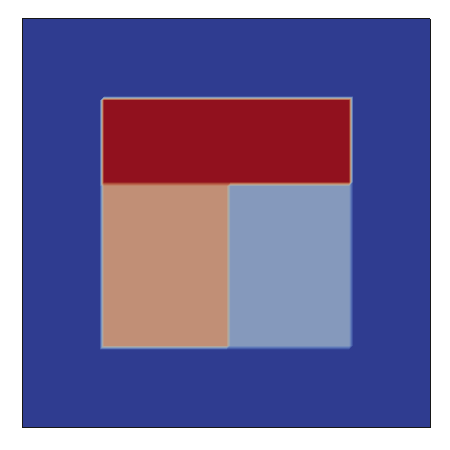}
  \label{fig:NCH-triple-init-color}
}\\
\subfloat[$\sigma_{ij}=1, t=1\times 10^{-4}$]{
  \includegraphics[width=0.25\textwidth]{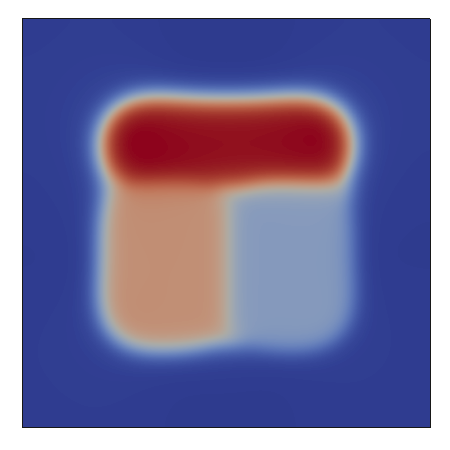} 
  \label{fig:NCH-triple-homo-1}
}%
\subfloat[$\sigma_{ij}=1, t=1\times 10^{-3}$]{
  \includegraphics[width=0.25\textwidth]{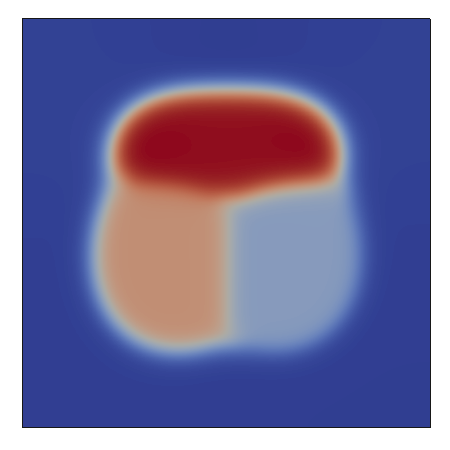} 
  \label{fig:NCH-triple-homo-2}
}%
\subfloat[$\sigma_{ij}=1, t=2\times 10^{-3}$]{
  \includegraphics[width=0.25\textwidth]{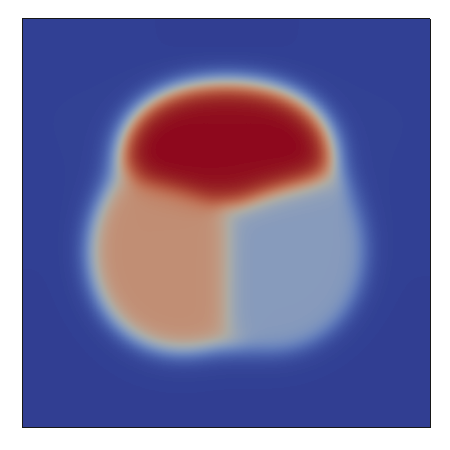}
  \label{fig:NCH-triple-homo-3}
}%
\subfloat[$\sigma_{ij}=1, t=8\times 10^{-3}$]{
  \includegraphics[width=0.25\textwidth]{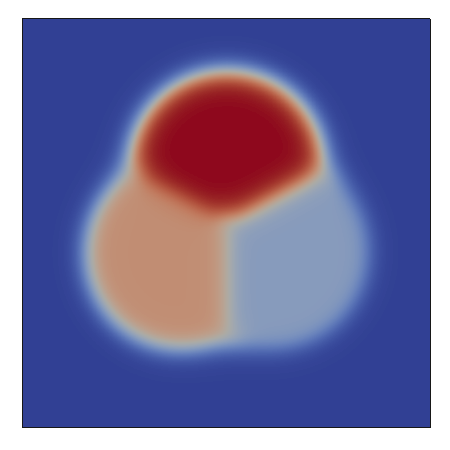}
  \label{fig:NCH-triple-homo-4}
}
\\
\subfloat[$\sigma_{12}=1.69, t=1\times 10^{-4}$]{
  \includegraphics[width=0.25\textwidth]{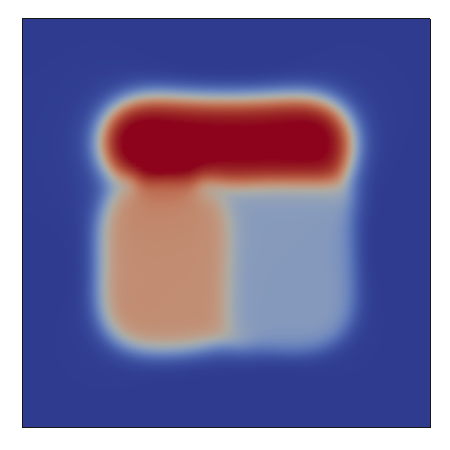} 
  \label{fig:NCH-triple-inhomo-1}
}%
\subfloat[$\sigma_{12}=1.69, t=1\times 10^{-3}$]{
  \includegraphics[width=0.25\textwidth]{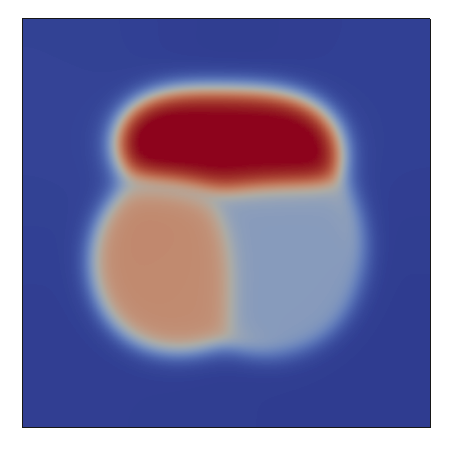} 
  \label{fig:NCH-triple-inhomo-2}
}%
\subfloat[$\sigma_{12}=1.69, t=2\times 10^{-3}$]{
  \includegraphics[width=0.25\textwidth]{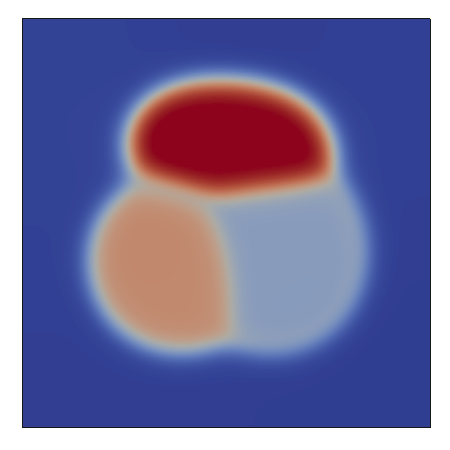}
  \label{fig:NCH-triple-inhomo-3}
}%
\subfloat[$\sigma_{12}=1.69, t=8\times 10^{-3}$]{
  \includegraphics[width=0.25\textwidth]{triple-inhomo1-3.png}
  \label{fig:NCH-triple-inhomo-4}
}
\\
\subfloat[$\sigma_{12}=2.56, t=1 \times 10^{-4}$]{
  \includegraphics[width=0.25\textwidth]{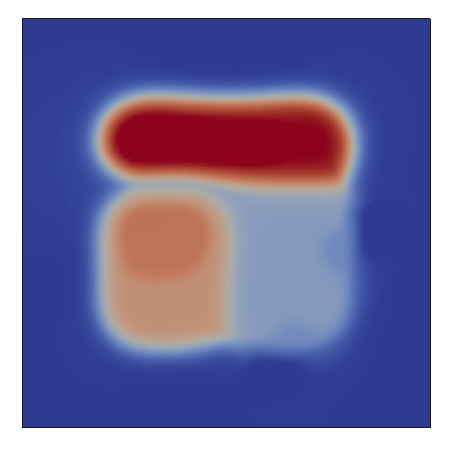} 
  \label{fig:NCH-triple-inhomo2-1}
}%
\subfloat[$\sigma_{12}=2.56, t=1 \times 10^{-3}$]{
  \includegraphics[width=0.25\textwidth]{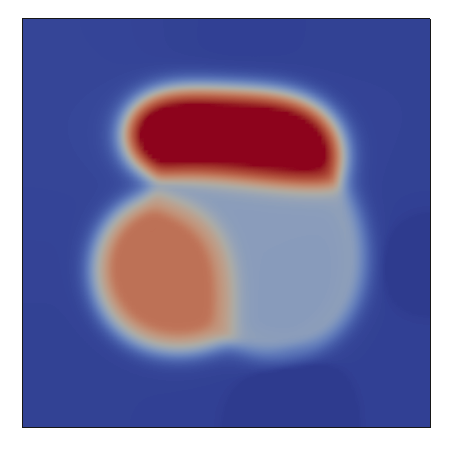} 
  \label{fig:NCH-triple-inhomo2-2}
}%
\subfloat[$\sigma_{12}=2.56, t=2 \times 10^{-3}$]{
  \includegraphics[width=0.25\textwidth]{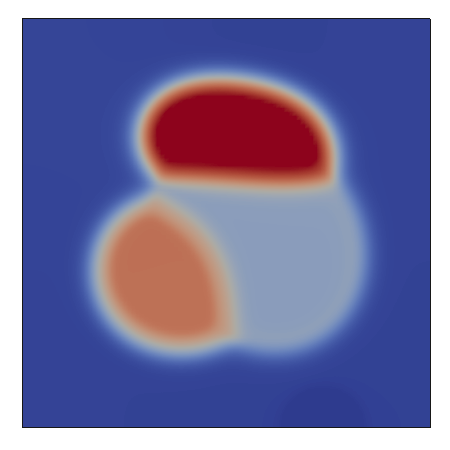}
  \label{fig:NCH-triple-inhomo2-3}
}%
\subfloat[$\sigma_{12}=2.56, t=8 \times 10^{-3}$]{
  \includegraphics[width=0.25\textwidth]{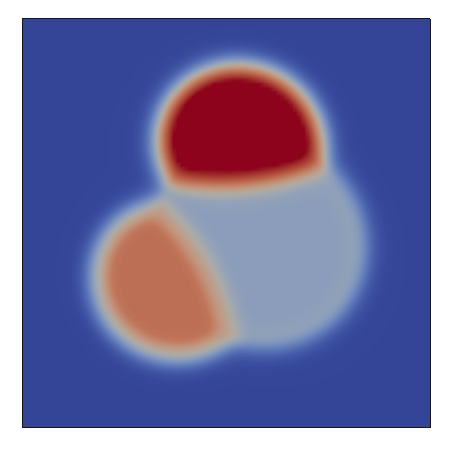}
  \label{fig:NCH-triple-inhomo2-4}
}
\caption{$N$-phase Cahn-Hilliard equations: Evolution of triple
junctions with different surface tensions}
\label{fig:NCH-triple}
\end{figure}

\section{Concluding Remarks} \label{sec:remarks}

In this paper, we presented multiphase Allen-Cahn and Cahn-Hilliard
models and their finite element discretizations accounting for the
effect of pairwise surface tensions. The free-energy functional with a
coefficient matrix in the capillary term was set up for the
generalized phase variables. By checking the consistency with the
two-phase model, we gave a set of linear equations between the
coefficient matrix and pairwise surface tensions. Thanks to the
relationship between the symmetric matrix space and simplex, we proved
the solvability of the coefficient matrix on the tangent space of
solution manifold --- an $(N-1)$-dimensional hyperplane. Furthermore,
we gave two sufficient and necessary conditions for the SPD of the
coefficient matrix --- conditions that are fundamental to the
well-posedness of $N$-phase Allen-Cahn and $N$-phase Cahn-Hilliard
models presented.

Our derivation of the $N$-phase Allen-Cahn and Cahn-Hilliard equations
stems from the formulation of the free-energy functional and the
gradient flows on the solution manifold. With the introduction of an
induced inner product on the tangent space, the dynamics of
concentrations of both models are inherently invariant, that is,
independent of the choice of phase variables. Based on this nice
property, a special choice of phase variables is used in the numerical
simulation to clarify the tangent space. 

We proposed semi-implicit, fully-implicit, and modified Crank-Nicolson
schemes in the finite element framework for $N$-phase Allen-Cahn
equations, such that the energy-stability properties are similar to
the two-phase model. We also numerically verified the efficiency and
energy-stability of each scheme by simulating the grain growth on the
unit square domain. For the finite element discretization of $N$-phase
Cahn-Hilliard equations, the semi-implicit, fully-implicit, and
modified Crank-Nicolson schemes were also discussed.  Further, the
effect of inhomogeneous surface tensions on the spinodal decomposition
was investigated.  Finally, we carried out numerical experiments
focused on the evolution of triple junctions in order to establish and
demonstrate the ability of these models to deal with inhomogeneous
surface tensions. 

\section*{Acknowledgements} 
The authors would like to express their gratitude to Prof. Chun Liu
and Dr. Yukun Li for their helpful discussions and suggestions, and to
thank the referees for the valuable comments leading to a better
version of this paper.

\section*{References}
\bibliographystyle{elsarticle-num} 
\bibliography{Multiphase}

\end{document}